\documentclass[draftclsnofoot,onecolumn,12pt,conference]{IEEEtran}

\usepackage{cite}
\usepackage{graphicx}
\usepackage{psfrag}
\usepackage{subfigure}
\usepackage{amsmath}
\usepackage{amssymb}
\usepackage{epsfig}
\usepackage{color}
\usepackage{epsf}
\usepackage{algorithmic}
\usepackage{algorithm}
\usepackage{url}

\newtheorem{theorem}{Theorem}

\newtheorem{lemma}[theorem]{Lemma}

\newtheorem{proposition}[theorem]{Proposition}

\newcommand{\Ns}{N_s}

\def\bb0{{\mathbb{0}}}

\def\bb{{\mathbf{b}}}

\def\bff{{\mathbf{f}}}
\def\bg{{\mathbf{g}}}
\def\bh{{\mathbf{h}}}

\def\bs{{\mathbf{s}}}

\def\bv{{\mathbf{v}}}
\def\bw{{\mathbf{w}}}
\def\bx{{\mathbf{x}}}
\def\by{{\mathbf{y}}}

\def\b0{{\mathbf{0}}}

\def\bF{{\mathbf{F}}}
\def\bG{{\mathbf{G}}}
\def\bH{{\mathbf{H}}}
\def\bI{{\mathbf{I}}}

\def\bK{{\mathbf{K}}}

\def\cH{\mathcal{H}}
\def\cI{\mathcal{I}}

\def\cV{\mathcal{V}}

\def\sf0{{\mathsf{0}}}

\def\rmC{\mathrm{C}}
\def\rmD{\mathrm{D}}

\def\rmI{\mathrm{I}}

\def\rmN{\mathrm{N}}

\def\rmP{\mathrm{P}}

\def\rmR{\mathrm{R}}

\def\rmV{\mathrm{V}}

\def\rmc{{\mathrm{c}}}

\def\rmt{{\mathrm{t}}}

\def\rm0{{\mathrm{0}}}

\def\Nt{{N_t}}
\def\Nr{{N_r}}

\def\Ns{{N_s}}

\DeclareMathOperator*{\argmax}{arg\,max}

\begin{document}
\title{Limited Feedback For Temporally Correlated MIMO Channels With Other Cell Interference\footnote{This work was supported by the Semiconductor Research Company (SRC) Global Research Consortium (GRC) task ID 1648.001.}}
\author{\authorblockN{Salam Akoum and Robert W. Heath, Jr.}
\authorblockA{Department of Electrical \& Computer Engineering \\
Wireless Networking and Communications Group\\
The University of Texas at Austin\\
1 University Station C0803 \\
Austin, TX 78712-0240 \\
\{salam.akoum, rheath\}@mail.utexas.edu\\}}
\maketitle

\begin{abstract}
Limited feedback improves link reliability with a small amount of feedback from the receiver back to the transmitter. In cellular systems, the performance of limited feedback will be degraded in the presence of other cell interference, when the base stations have limited or no coordination. This paper establishes the degradation in sum rate of users in a cellular system, due to uncoordinated other cell interference and delay on the feedback channel. A goodput metric is defined as the rate when the bits are successfully received at the mobile station, and used to derive an upper bound on the performance of limited feedback systems with delay. This paper shows that the goodput gained from having delayed limited feedback decreases doubly exponentially as the delay increases. The analysis is extended to precoded spatial multiplexing systems where it is shown that the same upper bound can be used to evaluate the decay in the achievable sum rate. To reduce the effects of interference, zero forcing interference cancellation is applied at the receiver, where it is shown that the effect of the interference on the achievable sum rate can be suppressed by nulling out the interferer. Numerical results show that the decay rate of the goodput decreases when the codebook quantization size increases and when the doppler shift in the channel decreases.
\end{abstract}


\section{Introduction} \label{sec:Intro}
Multiple input multiple output (MIMO) communication systems can use limited feedback of channel state information
from the receiver to the transmitter to improve the data rates and link reliability on the
downlink \cite{AGoldsmith,JeffWimax, OverviewLFWC}. With limited feedback, channel state information is quantized
by choosing a representative element from a codebook known to both the receiver and the transmitter.
Quantized channel state information is used at the transmitter to design intelligent transmission
strategies such as precoded spatial multiplexing and transmit beamforming \cite{LoveSpatialMultiplex, OverviewLFWC}.
Limited feedback concepts have been applied to more advanced system configurations such as MIMO-OFDM and multiuser MIMO and are proposed for current and next generation wireless systems \cite{OverviewLFWC}.

Most prior work on single user limited feedback MIMO focused on the block fading channel model, where the channel is assumed constant over one block and consecutive channel realizations are assumed independent. Following this assumption, limited feedback was cast as vector quantization problems \cite{VQbook}. Different methods for codebook design have been developed such as line packing \cite{LoveSpatialMultiplex, BeamFiniterateFeedback, Grassmanlove}, and Lloyd's algorithm \cite{Narula, Giannakis, RaoB, RVQ}. While these approaches are optimal for block-to-block independently fading channels, they do not capture the temporal correlation inherent in realistic wireless channels \cite{OverviewLFWC}. Feedback methods that can track the temporal evolution of the channel and adaptive codebook strategies are proposed to improve the quantization \cite{mondal2006channel}, \cite{yang2007transmission}. In \cite{mondal2006channel}, an adaptive quantization strategy in which multiple codebooks are used at the transmitter and the receiver to adapt to a time varying distribution of the channel is proposed. In \cite{yang2007transmission}, a new partial channel state information (CSI) acquisition algorithm that models and tracks the variations between the dominant subspaces of channels at adjacent time instants is employed. Markov models to analyze the effect of the channel time evolution and consequently, the feedback delay are proposed in \cite{FSMC},\cite{LFTCJ}, \cite{CapacityInterStreams}. Other temporal correlation models and measurement results of the wireless channel are used in \cite{FeedbackDelayModel2}, \cite{FeedbackDelayModel3}, \cite{BobGlobecomm} to evaluate the effect of the feedback delay. In \cite{FeedbackDelayModel2}, the authors quantize the parameters of the channel to be fed back using adaptive delta modulation, taking into consideration the composite delay due to processing and propagation. The authors in \cite{FeedbackDelayModel3} present measurement results of the performance of limited feedback beamforming when differential quantization methods are employed. Measurement results presented in \cite{BobGlobecomm} show that the upper bound on the throughput gain obtained using a Markov model, for an indoor wireless LAN setup, is accurate.

Feedback delay exists due to sources such as signal processing algorithms, propagation and channel access protocols. The effect of feedback delay on the achievable rate and bit error rate performance of MIMO systems has been investigated in several scenarios \cite{LFDelayConf, MarkovKroenecker, CapacityInterStreams , GiannakisPrediction, FiniteRateRao}. The
feedback delay has been found to reduce the achievable throughput \cite{LFDelayConf, MarkovKroenecker}, and to
cause interference between spatial data streams \cite{CapacityInterStreams}. Channel prediction was
proposed in \cite{GiannakisPrediction, KobayashiDelay} to remedy the effect of the feedback delay. Albeit not in the context of limited feedback, the authors use pilot symbol assisted modulation to predict the channel based on the Jakes model for temporal correlation of the channel. The authors in \cite{LFDelayConf, LFTCJ} derived expressions for the feedback bit rate, throughput gain and feedback compression rate as a function of the delay on the feedback channel.

Most of the works on limited feedback MIMO considered an ergodic metric for the achievable rate. Such a metric might not be appropriate to account for the slow-fading, temporally correlated channel. For limited feedback systems, the CSI at the transmitter gets corrupted due to errors on the feedback channel such as feedback delay, quantization, as well as noise. The uncertainty about the actual channel state causes the transmitted packets to become corrupted whenever the transmitted rate exceeds the instantaneous mutual information of the channel, hence causing an outage \cite{LauV},\cite{Aggarwal}, \cite{Kobayashi}. The effect of packet outage is accentuated in multi-cell environments, when the base stations have limited or no coordination. In the presence of uncoordinated other cell interference, the transmitter modulates its information at a rate that does not take into account the added interference at the mobile station, hence increasing the probability of outage. Thus, an analytical method that takes into account both delay and other cell interference is important to quantify the performance of limited feedback MIMO over slow fading channels, in multi-cell environments.

MIMO cellular systems are interference limited. While
multi-cell MIMO and base station cooperation techniques \cite{Shamai}, \cite{3GPPLTEA1}, \cite{3GPPLTEA2} can mitigate the effect of interference, when the base stations share full or partial channel state and/or data information,
they incur overhead that scales exponentially with the number of base stations. Issues such as complexity of joint processing across all the base stations, difficulty in acquiring full CSI from all the mobiles at each base station,
and time and phase synchronization requirements make full coordination extremely
difficult, especially for large networks. Thus while base station coordination is an attractive long term solution, in the near term, an understanding of the impact of the interference on limited feedback is required.
Single cell limited feedback MIMO techniques are expected to loose much of their effectiveness in the presence of multi-cell interference \cite{Sim_CellMIMO}. When each cell designs its channel state index independently of the other cell interference, a scenario where every transmitter-receiver pair is trying to optimize its own rate occurs, hence decreasing the overall sum rate of the limited feedback, when compared to noise limited environments.

In this paper, we derive the impact of delay on the achievable sum rate of limited feedback MIMO systems in the presence of other cell interference. To account for the packet outage, we use the notion of goodput. We define the goodput as the number of bits successfully transmitted to the receiver per unit of time, or in other words, the rate at the transmitter when it does not exceed the instantaneous mutual information of the channel. We model the fading of the MIMO channels in the system as independent first order finite state Markov chains. Assuming limited or no feedback coordination between adjacent base stations, and using Markov chain convergence theory, we show that the feedback delay, coupled with the other cell interference at the mobile station, causes the spectral efficiency of the system to decay exponentially. The decay rate almost doubles when the mobile station is at the edge of its cell, and hence interference limited.

We evaluate the joint effect of delay and uncoordinated other cell interference on the achievable sum goodput of both limited feedback beamforming and limited feedback precoded spatial multiplexing MIMO systems. We derive an upper bound on the goodput gain for both single stream and multi-stream limited feedback systems. We show that the goodput gain decays doubly exponentially with the feedback delay.
To mitigate the effect of other cell interference, while still assuming limited coordination between the base stations, we also consider the net performance improvement through the application of zero forcing interference cancellation at the receiver. Assuming one strong interferer and multiple antennas at the receiver, we use the available degrees of freedom to apply zero forcing (ZF) nulling \cite{FundWireless, InterferenceCellular}. We derive
closed form expressions of the achievable ergodic goodput with zero forcing cancellation, and compare its performance to that of the noise limited single cell environment.

The effect of delay on the throughput gain of a limited feedback beamforming system was considered in \cite{LFTCJ}. The authors, however, did not consider other cell interference nor did they account for the inherent packet outage. In contrast, our paper targets the performance of limited feedback systems in
interference limited scenarios and derives upper bounds of the performance limits of these systems as the delay on the feedback link increases.

The goodput notion we consider borrows from that considered in \cite{LauV}, \cite{Aggarwal}, \cite{Kobayashi}.  The authors in \cite{LauV} jointly design the precoders, the rate and the quantization codebook to maximize the achievable goodput of the limited feedback system. They consider the noise on the feedback channel as the only driver for the packet outage, and do not account for other cell interference. The authors in \cite{Kobayashi} use the goodput metric to design scheduling algorithms to combat the degradation in performance due to feedback delay in a single cell MIMO system setup. Similarly, the authors in \cite{Aggarwal} propose a greedy rate adaptation algorithm to maximize the goodput as a function of the feedback delay, using an automatic repeat request (ARQ) system to feedback channel state information. Practically, when the channel state information obtained through limited feedback is corrupted, one approach is to use fast hybrid automatic repeat request (HARQ) \cite{CostelloMiller}, to provide closed loop channel adaptation.  ARQ, however, is also subject to delay and errors on the feedback channel. In this paper, we assume in our definition of goodput that the HARQ is not present, and hence packets transmitted at a rate higher than the instantaneous rate of the channel will be lost.

\textbf{Organization:} This paper is organized as follows. In Section \ref{sec:sysModel},
we describe the limited feedback multicell system considered. In Section \ref{sec:LFfeedback},
we present the limited feedback mechanism employed. Section \ref{sec:SysGoodput} introduces the
 system goodput, and Section \ref{sec:secMarkov} describes the channel state Markov chain used
 for the analysis. In Section \ref{sec:analysis}, we present the effect of the feedback delay
 on the goodput gain. Section \ref{sec:InterfCancel} presents ZF interference cancellation
 at the receiver to mitigate the effect of other cell interference. Section \ref{sec:Precoding} extends the
 results in Section \ref{sec:analysis} to precoded spatial multiplexing. Section
 \ref{sec:results} presents numerical results that show the different aspects of
 the relation between the feedback rate gain and the feedback delay. This is
 followed by concluding remarks in Section \ref{sec:Conclusion}.

\textbf{Notation:} Bold lowercase letters $\mathbf{a}$ are used to denote column vectors, bold uppercase letters $\mathbf{A}$ are used to denote matrices, non bold letters $a$ are used to denote scalar values, and caligraphic letters $\mathcal{A}$ are used to denote sets or functions of sets. Using this notion, $|a|$ is the magnitude of a scalar, $\|\mathbf{a}\|$ is the vector 2-norm, $\mathbf{A}^*$ is the conjugate transpose, $\mathbf{A}^T$ is the matrix transpose, $[\mathbf{A}]_{lm}$ is the scalar entry of $\mathbf{A}$ in the $\ell^{th}$ row and the $k^{th}$ column. We use $\mathbb{E}$ to denote expectation and $\mathrm{a}^{\mathrm{t}}$ to denote the metric $\mathrm{a}$ evaluated at the transmitter.

\section{System Model} \label{sec:sysModel}

We consider the modified Wyner type \cite{Wyner94} $N_B$-cell $K$-user per cell circular array cellular model. The base stations $B_i,\; i = 1,\cdots, N_B$, with $N_t$ transmit antennas each, serve mobile stations $M_i$ with $N_r$ receive antennas.
We index the mobile users by the same index of the base station they receive their desired signal from, for tractability.
The users are located at the edge of their cells, such that each user is reachable from the two closest base stations only. The base stations have limited or no coordination. Figure \ref{fig:cell_sys} illustrates the cellular model for two adjacent interfering base stations.

Each cell employs a limited feedback beamforming system. The system, illustrated in Figure \ref{fig:LF_sys}, is discrete time, where continuous time signals are sampled at the symbol rate $1/T_s$, with $T_s$ being the symbol duration. Consequently, each signal is represented by a sequence of samples with $n$ denoting the sample index. Assuming perfect synchronization between the base stations, matched filtering, and a narrowband channel, the $n$-th received data sample $\by_1[n]$ for a single user of interest $M_1$ in base station $B_1$ can be written as
\begin{eqnarray}
\nonumber \by_1[n] &=& \sqrt{\frac{\alpha_1}{\Nt}}\bH_1[n]\bx_1[n] + \sqrt{\frac{\alpha_2}{\Nt}}\bG_2[n]\bx_2[n]+ \bv_1[n],
\end{eqnarray}
where $\by_1[n] \in \mathbb{C}^{N_r\times 1}$ is the received signal vector at $M_1$; $\bH_1[n] \in \mathbb{C}^{N_r\times N_t}$ is the small-scale fading channel between $B_1$ and $M_1$.
$\alpha_1$ and $\alpha_2$ are the received powers of the desired and interfering signal, respectively,  at $M_1$.
$\bG_2[n] \in \mathbb{C}^{N_r\times N_t}$ represents the
$n$-th realization of the MIMO channel between $B_2$ and $M_1$. $\bx_1[n]$ is the desired transmit signal vector for $M_1$, subject to the power constraint
$\mathbb{E}[\|\bx_1\|^2]= N_t$.  If transmit beamforming is employed, and hence only one stream $s_1[n]$ is transmitted at time $n$, the signal $\bx_1[n] = \bff_1[n]s_1[n]$, where $\bff_1[n]$ is the unit norm beamforming vector. If, however, precoded spatial multiplexing is used, $\bx_1[n]=\bF_1[n]\bs_1[n]$, where $\bF_1[n]$ in $\mathbb{C}^{N_t\times N_s}$
is unitary ($\bF^*_1\bF_1 = \frac{1}{N_s}\bI_{N_s}$), $N_s$ is the number of spatial multiplexing streams transmitted.
$\bx_2[n]$ is the interfering transmit signal vector designated for $M_2$ served by base station $B_2$, subject to the power constraint
$\mathbb{E}[\|\bx_2\|^2]= N_t$; $\bv_1[n] \in \mathbb{C}^{Nr\times 1}$ is $\mathcal{CN}(0, \bI)$, modeling the additive noise observed at $M_1$.

The random processes $\{\bH_i[n]\}$ and $\{\bG_i[n]\}$ are assumed stationary, ergodic and temporally correlated. The assumption that these channels are Gaussian distributed is not necessary for our analysis.

Moreover, the desired and interfering channels at $M_1$, $\bH_1[n]$ and $\bG_2[n]$, are independent, since the base stations are geographically separated.
They are assumed to be perfectly known at $M_1$, thereby ignoring channel estimation error at the receiver.

\section{CSI Limited Feedback}\label{sec:LFfeedback}
In this paper, we consider a finite rate feedback link, as depicted in Figure \ref{fig:LF_sys}. The mobile user $M_i$ first estimates the channel state information sequence $\{\bH_i[n]\}$ using pilot symbols sent by the base station $B_i$. Next, the CSI quantizer efficiently quantizes the channel sequence by means of a Grassmannian codebook, as
outlined in  \cite{LoveSpatialMultiplex,Grassmanlove}. The quantization process depends on whether transmit beamforming or precoded spatial multiplexing is used, as outlined in the following subsections. The quantization index is sent to the transmitter via a limited feedback channel.

\subsection{Transmit beamforming}
For the case of transmit beamforming, the beamforming vector $\bff$ has rank one, $\bff \in \mathcal{C}^{N_t\times 1}$, and the base stations send a one dimensional stream of data $s$ to the mobile users. To maximize the signal to noise ratio (SNR) for a given channel realization $\bH[n]$, the quantizer function $\mathcal{Q}$ at the receiver maps the channel matrix $\bH[n]$ to beamforming vector $\bff[n]$ in the codebook $\mathcal{F}$ and a corresponding index $I_n$ such that
\begin{eqnarray}\label{eqn:quantize_beam}
\bff[n] &=& \mathcal{Q}\{\bH[n]\} = \argmax_{\bv_l\in\mathcal{F}}\|\bH[n]\bv_l\|^2, \quad 1\leq l\leq N.
\end{eqnarray}
The channel $\bH[n]$ is mapped to the index $I_n = \ell$ if the code $\bv_l$ maximizes the SNR metric $\|\bH[n]\bv_l\|^2$.
$\bH[n]$ is then said to be in the Voronoi cell $\mathcal{V}_l$.
The feedback of the index $I_n$, called the feedback state, is sufficient for the transmitter to choose the necessary beamforming vector from the same codebook $\mathcal{F}$. The feedback state requires $\log_2(N)$ bits, where $N$ is the number of possible codes in the codebook.

\subsection{Precoded spatial multiplexing}
For the case of precoded spatial multiplexing, the precoder $\bF$ is a unitary matrix with rank $N_s$, where $N_s$ is the number of spatial multiplexing streams. Several different criteria are available to choose the optimal precoding matrix from a given codebook \cite{LoveSpatialMultiplex}. We choose the mutual information maximizing criterion, where the quantizer function $\mathcal{Q}$ maps the channel matrix to the precoding matrix $\bF[n]$ that maximizes the mutual information expression
\begin{eqnarray}\label{eqn:quantize_precode}
\nonumber \bF[n] &=& \mathcal{Q}\{\bH[n]\} = \argmax_{\bF_l\in\mathcal{F}}\bI(\bF_l) =\argmax_{\bF_l\in\mathcal{F}}\log_2\left(\det\left(\bI_{N_s}+ \frac{1}{N_s}(\bH[n]\bF_l)^*\bH[n]\bF_l\right)\right), \quad 1\leq l\leq N.\\
\end{eqnarray}
The receiver then sends the precoding matrix index $I_n = \ell$, corresponding to the precoder $\bF_l$, to the base station.

We assume that the feedback channel is free of error but has a delay of D samples. The error free assumption is
justifiable as the control channels are usually protected using aggressive error correction coding. Given this feedback channel,
the channel state information available at the transmitter $I_{n-\mathrm{D}}$ lags behind the actual channel state $I_n$ at the receiver by D samples.
The delay is primarily caused by signal processing algorithms complexity, channel access protocols and propagation delay. A fixed delay D is assumed on the feedback channel in all base stations. Since the propagation delay has little contribution to the total amount of delay and the other sources of delay, caused by processing at the receiver, are similar across users, different users in different cells experience the same amount of delay.
In what follows, we study the performance of limited feedback MIMO systems over temporally correlated channels for the case of transmit beamforming in Section \ref{sec:analysis}. We extend the analysis to precoded spatial multiplexing in Section \ref{sec:Precoding}.


\section{The Goodput Expression over the Markov Channel Model}
The channel state information that reaches the transmitter suffers from feedback delay and quantization error. Moreover, the CSI quantizer at the receiver chooses the quantization codeword to maximize the desired signal to noise ratio, without taking into account the interference from the neighboring base station. Consequently, the CSI at the transmitter does not contain any information about the other cell interference affecting the mobile user. The base station, assuming that its received CSI is accurate, modulates its transmit signal at a rate corresponding to its erroneous CSI, sometimes resulting in a \emph{rate outage} or \emph{packet outage}, when the transmit rate exceeds the instantaneous mutual information of the channel. In this paper, we assume that HARQ is not present, and we evaluate the amount of information received without error as a function of the delay on the feedback channel.

\subsection{Conditional System Goodput}\label{sec:SysGoodput}
To account for the rate outage, we assume that any transmission at a rate higher than the capacity of the channel
fails. In other words, if the rate at the transmitter $\mathrm{R}^{\mathrm{t}}[n-\mathrm{D}]$, where D is the delay on the feedback channel, exceeds the instantaneous mutual information at the receiver $\mathrm{R}[n]$ , the transmission is declared unsuccessful.
The instantaneous system goodput $\rho[n,\mathrm{D}]$ is defined as
\begin{eqnarray}
\rho[n, \mathrm{D}] &=& \mathrm{R}^{\mathrm{t}}[n-\rmD]\mathrm{ }\mathcal{I}\left(\mathrm{R}^{\mathrm{t}}[n-\rmD]\leq \mathrm{R}[n]\right),
\end{eqnarray}
where $\mathcal{I}(A)$ is the indicator function, which evaluates to $1$ if the event $A$ is true, and $0$ otherwise\cite{LauV}.
The ergodic goodput, averaged over the set $\large{\cH}$ of the MIMO fading channels,
$$\large{\mathcal{H}} = \{\bH_1[n],\bH_1[n-\rmD], \bH_2[n], \bH_2[n-\rmD], \bG_2[n]\},$$
can be expressed as
\begin{equation}
\bar{\rho}(\rmD) = \mathbb{E}_{\large{\cH}}\left[\mathrm{R}^{\mathrm{t}}[n-\rmD] \mathbb{P}\left(\mathrm{R}^{\mathrm{t}}[n-\rmD] \leq \mathrm{R}[n] \right)\right].
\end{equation}
In the sequel, we approximate the fading as a discrete time Markov process and use the Markov structural properties to derive closed form expressions of the ergodic goodput, as a
function of the feedback delay, with and without interference cancellation methods at the receiver.
\subsection{Channel State Markov Chain}\label{sec:secMarkov}
We approximate the fading of the MIMO channels $\bH_k$  as a discrete time first order Markov process \cite{FSMC}.
Since the feedback state index $I_{n}$ is mapped from the channel $\bH_k[n]$ by the quantization function in (\ref{eqn:quantize_beam}) and (\ref{eqn:quantize_precode}), we follow the approach in \cite{MarkovModelsK}, and we model the time variation of the feedback state $I_{n}$ by a first order finite state Markov chain. This Markov chain $\{I_{n}\}$, mapped by the quantizer function from a stationary channel $\bH[n]$, is stationary and has the finite state space $\mathcal{I} = \{1,2,3,\cdots,N\}$, where $N$ is the size of the codebook. The states of this Markov chain are one-to-one mapped with the Voronoi cells $\mathcal{V}_i$ of the channel matrices $\bH[n]$. The probability of transition from state $I_{n} = m$ to state $I_{r} = \ell$ is given by $P_{ml}$. The stochastic matrix is thus $\mathbf{P}$,
with $[\mathbf{P}]_{ml} = P_{ml}$.

The Markov chain is assumed ergodic with a stationary distribution
vector $\boldsymbol{\pi}$, where $\mathbb{P}\left(I_{n} = i\right)= \pi_i$. The stationary probabilities $\pi_i$ are assumed equal, $\pi_i = \pi= \frac{1}{N}$. This follows from the fact that the stationary probability of each Markov state is proportional to the area of its corresponding Voronoi cell, and the Voronoi regions for the codes in the codebook are assumed to have equal volume \cite{MarkovModelsK}.

For the two cell system under investigation, the random processes $\bH_k[n]$, $\bG_k[n]$ in different cells are assumed independent. The base stations have limited or no coordination hence the joint probability mass function of the random processes in the set $\large{\cH}$ is given by the product of the probability mass functions of the processes in each cell individually. For the desired base station, the joint probability mass function between $\bH_1[n]$ and $\bH_1[n-\rmD]$ is given by
\begin{equation}\label{eqn:prtransi}
\mathbb{P}\left[\bH_1[n] \in \mathcal{V}_{k_{10}}, \bH[n-\rmD] \in \mathcal{V}_{k_{1D}}\right]
= \left[\mathbf{P}^{\rmD}\right]_{k_{1D}k_{10}}\pi_{k_{1D}} =  \left[\mathbf{P}^{\rmD}\right]_{k_{1D}k_{10}}\pi.
\end{equation}
$\mathbf{P}^{\rmD}$ represents $\mathbf{P}$ to the
power of D.  The indices $k_{1i}$ are used to denote the base station $1$ and the amount of delay in time samples $i$. In general, $\left[\mathbf{P}^{\rmD}\right]_{ij}$ does not yield closed form expression for the Markov chain probabilities except for the case of single antennas. The channels corresponding to the two different base stations are independent, and hence the individual Markov chain
transition probabilities are independent but identically distributed.  These probabilities are computed
by Monte Carlo simulations as shown in Section \ref{sec:results}.

\section{The Effect of Delay on the Feedback Goodput Gain}\label{sec:analysis}
We consider the effect of fixed feedback delay on the average system goodput of
the MIMO interference system. The delay D on the feedback channel for the interfering
as well as the desired cell is considered fixed, caused by signal processing,
propagation delay and channel access control.

The instantaneous mutual information computed at the receiver, in the presence of other cell interference, is expressed as
\begin{equation}
\mathrm{R}[n] = \log_2(1 + \mathrm{SINR}[n]),
\end{equation}
where the signal to interference noise ratio $\mathrm{SINR}[n]$, assuming MRC combining at the mobile stations, is given by
\begin{equation}\label{eqn:rate_RX}
\mathrm{SINR}[n] = \frac{\alpha_1\|\bH_1[n]\bff_1[n-\rmD]\|^2}{{\alpha_2\|\bw^*[n]\bG_2[n]\bff_2[n-\rmD]\|^2} + \Nt\|\bw^*[n]\bv_1[n]\|^2},
\end{equation}
where $\bw[n] = \frac{\bH_1[n]\bff_1[n-\rmD]}{\|\bH_1[n]\bff_1[n-\rmD]\|}$ is the MRC combining vector. In general, the $\mathrm{SINR}$ distribution does not have a closed form expression and has to be estimated using Monte Carlo simulations.

At the base station, the transmitter modulates its signal based on the delayed CSI $\rmI_{n-\rmD}$. Assuming continuous rate adaptation and Gaussian transmit signals, the instantaneous transmission rate depends on the delayed precoder $\bff_1[n-\rmD]$ and its corresponding channel $\bH_1[n-\rmD]$, i.e.,
\begin{equation}\label{eqn:rate_TX}
\mathrm{R}^{\mathrm{t}}[n-\rmD] = \log_2(1 + \mathrm{SINR}^{\mathrm{t}}[n]) = \log_2(1 + \|\bH_1[n-\rmD]\bff_1[n-\rmD]\|^2).
\end{equation}
We write the average system goodput $\bar{\rho}(\rmD)$, based on (\ref{eqn:rate_RX}) and (\ref{eqn:rate_TX}),
\begin{eqnarray}\label{eqn:ergodGood}
\nonumber \bar{\rho}(\rmD) &=& \mathbb{E}_{\cH}\left[\mathrm{R}^{\mathrm{t}}[n-\rmD]\cI\left(\mathrm{R}[n]\geq \mathrm{R}^{\mathrm{t}}[n-\rmD]\right)\right]\\
&=&\mathbb{E}_{\bH_1[n-\rmD]}\left[\mathrm{R}^{\mathrm{t}}[n-\rmD]\mathbb{P}\left(\mathrm{R}[n]\geq \mathrm{R}^{\mathrm{t}}[n-\rmD]\;\large|\; \bH_1[n-\rmD]\right)\right]\\
\nonumber &=& \mathbb{E}_{\bH_1[n-\rmD]}\left[\mathrm{R}^{\mathrm{t}}[n-\rmD]\mathbb{P}\left(\mathrm{SINR}[n]\geq \mathrm{SINR}^{\mathrm{t}}[n]\;\large|\;\bH_1[n-\rmD]\right)\right].
\end{eqnarray}
Following (\ref{eqn:prtransi}), $\bar{\rho}(\rmD)$ can be obtained using the transition probabilities of the Markov chains of the desired and interfering channels.

We distinguish between two types of interference, \emph{severe interference} and \emph{mild interference}. Severe interference occurs when the interference channel $\bG_2$ falls in the same Voronoi cell as that of the transmit beamforming vector $\bff_2$, i.e.,  $\bff_2$ is chosen to maximize $\|\bG_2\bff_2\|^2$. In this case, the interference channel is mapped to a state that corresponds to the beamforming vector that maximizes $\|\bH_2\bff_2\|^2$ in the feedback index Markov chain.
Otherwise, the interference is considered to be mild. The probability $r$ of the interference being severe depends on the probability that the interference channel falls in the same subspace as that of the channel $\bH_2$ and hence, on the size $N$ of the codebook used to quantize the channels. For the case  of severe interference, the probability
\begin{equation}
\mathbb{P}\left[\bG_2[n] \in \mathcal{V}_{k_{20}}, \bH_2[n-\rmD] \in \mathcal{V}_{k_{2D}}\right]
= \left[\mathbf{P}^{\rmD}\right]_{k_{2D}k_{20}}\pi_{k_{2D}} =  \left[\mathbf{P}^{\rmD}\right]_{k_{2D}k_{20}}\pi.
\end{equation}
Consequently, the ergodic goodput gain is
\begin{eqnarray}
\nonumber \bar{\rho}_1(\rmD) &=& \mathbb{E}_{\bH_1[n-\rmD]}\left[\mathrm{R}^{\mathrm{t}}[n-\rmD]\mathbb{P}\left(\mathrm{R}^{\mathrm{t}}[n-\rmD] \leq \mathrm{R}[n,\rmD] \;\large|\;\bH_1[n-\rmD]\right)\right]\\
\nonumber &=& \sum_{k_{1D}=1}^{N}{\rmR^{\rmt}_{k_{1D}}\mathbb{P}\left(\mathrm{R}^{\mathrm{t}}_{k_{1D}} \leq \mathrm{R}[n,\rmD]\;\large|\;\bH_1[n-\rmD]\in \mathcal{V}_{k_{1D}} \right)\mathbb{P}(\bH_1[n-\rmD] \in \mathcal{V}_{k_{1D}})}\\
              &=&\sum_{k_{10}, k_{1D}}{\sum_{k_{20}, k_{2D}}{\mathrm{C}(\rmD)\left[\mathbf{P}^{\rmD}\right]_{k_{1D}k_{10}}\left[\mathbf{P}^{\rmD}\right]_{k_{2D}k_{20}}\pi^{2}}},
\end{eqnarray}
where $\rmR^{\rmt}_{k_{1D}}$ is the rate at the transmitter, $\rmR^{\rmt}$, with $\bH_1[n-\rmD] \in \cV_{k_{1D}}$.
Similarly, $\mathrm{R}_{k_{kD}k_{k0}}$ is the instantaneous rate at the receiver, $\mathrm{R}[n,\rmD]$, with $\bH_k[n-\rmD] \in \mathcal{V}_{k_{kD}}, \bH_k[n] \in \mathcal{V}_{k_{k0}}$,  $k \in \{1,2\}$.
We obtain $\rmC(\rmD)$ as $\mathrm{C}(\rmD) = \left(\mathrm{R}^{\mathrm{t}}_{k_{1D}}\mathbb{P}\left(\mathrm{R}_{k_{kD}k_{k0}} \geq \mathrm{R}^{\mathrm{t}}_{k_{1D}} \right)\right)$.

For the case of mild interference, when the codebook index does not maximize $\|\bG_2\bff_2\|^2$, the joint probability mass function of the random processes $\bG_2$ and $\bff_2$ is given by the product of the stationary probability of each random process individually. Hence, $\bar{\rho}_2$ is simply computed by
\begin{eqnarray}
\bar{\rho}_2(\rmD) &=& \sum_{k_{10}, k_{1D}}{\sum_{k_{20}, k_{2D}}{\rmC(\rmD)\left[\mathbf{P}^{\rmD}\right]_{k_{1D}k_{10}}\pi^{3}}}.
\end{eqnarray}
Finally, the ergodic goodput $\bar{\rho}(\rmD)$ is written in terms of the probability of the interference being severe.
\begin{eqnarray}
\nonumber \bar{\rho}(\rmD)&=& \mathbb{P}(\mathrm{severe}\; \mathrm{interference})\bar{\rho}_1(\rmD) + \left(1-\mathbb{P}(\mathrm{severe}\; \mathrm{interference})\right)\bar{\rho}_2(\rmD)\\
&=& r\bar{\rho}_1(\rmD) + \left(1-r\right)\bar{\rho}_2(\rmD).
\end{eqnarray}
Given a quantization size $N$ of the channel space, the probability of a random channel matrix falling in a Voronoi region $\cV_k$ is $\frac{1}{N}$. Hence, the probability of the interference channel $\bG_2$ falling in the Voronoi cell $\cV_{l}$ pertaining to the quantization vector $\bff_2$ is $\frac{1}{N}$. In other words, the probability of the interference being severe is $r = \frac{1}{N}$.

To capture the effect of increasing feedback delay on the ergodic system goodput, we use the notion of \emph{throughput gain}, defined in \cite{LFTCJ} as the throughput with delay D minus the throughput when the delay goes to infinity. When the delay goes to infinity, the feedback information becomes obsolete and thus irrelevant. The goodput gain is formally written as
\begin{equation}\label{eqn:DeltaR}
\Delta\bar{\rho}(\rmD) = \bar{\rho}(\rmD) - \bar{\rho}(\infty),
\end{equation}
where $\bar{\rho}(\infty)$ is given by
\begin{eqnarray}
 \bar{\rho}(\infty)  &=& 
               \sum_{k_{10},k_{1D}}^N\sum_{k_{20},k_{2D}}^N{\rmC(\rmD)\pi^{4}}.
\end{eqnarray}
This follows from the fact that, as $\rmD \rightarrow \infty$, the channel state Markov chain converges to the stationary distribution
\begin{equation}
\mathbf{P}^{\rmD} \rightarrow \left[\boldsymbol{\pi}, \cdots, \boldsymbol{\pi}\right].
\end{equation}
The goodput gain allows us to analyze the effect of increasing feedback delay on the multicell system,
and draw conclusions as to when closed loop limited feedback MIMO systems are feasible.
We derive an upper bound on the ergodic goodput gain $\Delta\bar{\rho}(\rmD)$, based on the Markov chain convergence
rate \cite{MarkovChains, MCconvergence}.

For this reason, we invoke theorem 2.1 in \cite{MCconvergence}, to upperbound the goodput gain $\Delta\bar{\rho}(D)$ in terms of the properties of the stochastic matrix of the channel Markov chain. Theorem 2.1 in \cite{MCconvergence} states that for the ergodic channel state Markov chain, the following inequality holds \begin{equation}\label{eqn:thm_conv}
\left(\sum_{m=1}^{N}|\left[\mathrm{P}^{\rmD}\right]_{lm}-\pi_m|\right)^2 \leq \frac{\lambda^{\rmD}}{\pi_l},\quad\quad 1\leq l \leq N,
\end{equation}
where $\lambda \in [0,1]$ is the second largest eigenvalue of the matrix $\mathbf{P}\tilde{\mathbf{P}}$.
The matrix $\tilde{\mathbf{P}}$ is defined as the time reversal of the stochastic matrix $\mathbf{P}$.
%
\begin{proposition}\label{prp:prop_1}
For fixed feedback delay of D samples, the feedback goodput gain can be bounded as
\begin{equation}\label{eqn:bound}
\Delta\bar{\rho}(\rmD) \leq a\left(\sqrt{\lambda}\right)^{2\rmD} + b\left(\sqrt{\lambda}\right)^{\rmD},
\end{equation}
where $a =  r\sum_{k_{1D},k_{2D}}{\max_{k_{10},k_{20}}{\mathrm{C}(\rmD)}\pi}$, and
$$b = r\sum_{k_{1D},k_{10},k_{2D}}{\max_{k_{20}}{\mathrm{C}(\rmD)}\pi^{2}\sqrt{\pi}} + \sum_{k_{1D}, k_{2D},k_{20}}{\max_{k_{10}}{\mathrm{C}(\rmD)}\pi^{2}\sqrt{\pi}}.$$
\end{proposition}
\begin{proof}
See Appendix. 
\end{proof}
The coefficients $a$ and $b$ depend on the instantaneous rate $\rmC(\rmD)$ and the stationary probability distribution $\pi$ of the Markov chain. The coefficient $a$ decreases with $r$ such that, as the codebook size increases, $a\rightarrow 0$. This causes the rate of decay of the goodput gain to approach that of the noise limited environment \cite{LFTCJ}. In that case, the coefficient $b \rightarrow \sum_{k_{1D}, k_{2D},k_{20}}{\max_{k_{10}}{\mathrm{C}(\rmD)}\pi^{2}\sqrt{\pi}}$.

We make the following observations about the conclusions in Proposition \ref{prp:prop_1}.
\begin{enumerate}
\item The feedback gain decreases at least exponentially with the feedback delay. The decreasing rate is $\lambda$ or $\sqrt{\lambda}$, depending on the values of the coefficients $a$ and $b$. The rate is thus determined by the channel coherence time and the size of the codebooks used.
\item The eigenvalue $\lambda$ is a key parameter in characterizing the behavior of the system. A larger value of $\lambda$ indicates longer channel coherence time and larger codebook size and vice versa.
\item The coefficients $a$ and $b$ depend on the number of precoders used to quantize the channels spaces. The coefficient $a$ is largely governed by $r$ which increases with decreasing $N$. As the delay grows larger, for values higher than $D=10$ time samples for example, the term $(\sqrt{\lambda})^{2D}$ decreases faster than $(\sqrt{\lambda})^{D}$. This means that the second term, and thus the exponential rate $(\sqrt{\lambda})^{\rmD}$ becomes dominant and the exponential decay in the throughput gain is brought closer to that of the single cell environment.
\item Simulation results in Section \ref{sec:results} show that the upper bound, derived in proposition \ref{prp:prop_1} is tight for several cases of interest.
\end{enumerate}

%
\section{Interference cancellation at the receiver}\label{sec:InterfCancel}
The results in Section \ref{sec:analysis} suggest that the other cell interference almost doubles the exponential rate of decrease of the goodput gain, when compared with the rate of decrease computed in \cite{LFTCJ}. We propose using interference cancellation techniques at the receiver, to compensate for the rate decay. Interference cancellation techniques are suitable when the base stations have limited coordination. We require the receiver to learn both its desired effective channel and the effective interference channel from the other cell base station. We implement zero forcing (ZF) interference cancellation.

Assuming the number of receive antennas $\Nr \geq 2$, we rewrite the received signal at $M_1$ as
\begin{eqnarray}
\nonumber \by_1[n] &=& \sqrt{\frac{\alpha_1}{\Nt}}\bH_1[n]\bff_1[n-\rmD]s_1[n] + \sqrt{\frac{\alpha_2}{\Nt}}\bG_2[n]\bff_2[n-\rmD]s_2[n] + \bv_1[n]\\
&=& \sqrt{\frac{\alpha_1}{\Nt}}\bh_1[n]s_1[n] + \sqrt{\frac{\alpha_2}{\Nt}}\bg_2[n]s_2[n] + \bv_1[n].
\end{eqnarray}
where $\bh_1[n] \in \mathbb{C}^{\Nr\times 1}$ is the effective desired channel at $M_1$ from $B_1$. $\bg_2[n] \in \mathbb{C}^{\Nr\times 1}$ is the effective interference from $B_2$.
The zero forcing linear receiver effectively projects the desired signal $\bh_1[n]$ onto the subspace orthogonal to the subspace of the interference $\bg_2[n]$. Hence, the resulting signal power at the receiver is
\begin{equation}
\gamma[n]=\|\bh_1[n]\|^2|\sin(\theta_n)|^2,
\end{equation}
where $\theta_n = \angle(\bh_1[n], \bg_2[n])$ is the angle between the $\bh_1[n]$ and $\bg_2[n]$.
$|\sin(\theta_n)|^2$ can be viewed as a projection power loss factor for the zero forcing receiver.
And the corresponding instantaneous rate at the receiver is expressed as
\begin{eqnarray}
\rmR[n] = \log_2(1 + p\gamma[n]),
\end{eqnarray}
where $p = \frac{\alpha_1}{\Nt N_0}$ is the signal to noise ratio at the receiver. Thus the
instantaneous system goodput is
\begin{equation}
\rho[n] = \rmR^{\rmt}[n-\rmD]\mathbb{P}(\rmR^{\rmt}[n-\rmD] \leq \log_2(1 + p\gamma[n])).
\end{equation}
The ergodic goodput follows, similar to (\ref{eqn:ergodGood}),
\begin{eqnarray}
\nonumber \bar{\rho}(\rmD) &=& \mathbb{E}_{\bH_1[n-\rmD]}\left[\rmR^{\rmt}[n-\rmD]\mathbb{P}\left(\rmR^{t}[n-\rmD] \leq \rmR[n]\;\large|\;\bH_1[n-\rmD]\right)\right]\\
&=& \mathbb{E}_{\bH_1[n-\rmD]}\left[\rmR^{\rmt}[n-\rmD]\mathbb{P}\left(\gamma[n] \geq\frac{\|\bH_1[n-\rmD]\bff_1[n-\rmD]\|^2}{p} \;\large|\;\bH_1[n-\rmD]\right)\right].
\end{eqnarray}
The key to computing $\bar{\rho}(\rmD)$ is to evaluate the cumulative density function (cdf) of $|\sin(\theta_n)|^2$.
This will allow us to express the probability of outage as a function of the channel gains only. We propose the following lemma on the distribution of $|\sin(\theta_n)|^2$.
\begin{lemma}
The squared norm of the sine of the angle between the effective desired and interference channel,  in a limited feedback beamforming system, when the channels $\bH_1[n]$ and $\bG_2[n]$ have i.i.d. complex Gaussian entries with zero mean and unit variance, can be approximated by a beta distribution with parameters $\Nr-1$ and $1$, on the interval $[0,\;1]$,
\begin{eqnarray}
\nonumber f_{|\sin(\theta)|^2}(y) &=& \frac{1}{\beta(\Nr-1,1)}y^{\Nr -2}(1-y)^{1-1}, \quad 0\leq y\leq 1\\
&=& (M-1)y^{M-2},\quad 0\leq y\leq 1.
\end{eqnarray}
where $\beta(\Nr-1,1) = \frac{\Gamma(\Nr)}{\Gamma(\Nr-1)\Gamma(1)}$ is the beta function, defined in termed of the Gamma function $\Gamma$, and the effective channels are assumed Gaussian distributed.
\end{lemma}
\begin{proof}
The effective channel vectors $\bh_1[n]$ and $\bg_2[n]$ are two independent random vectors whose entries, $\bh_{1i}[n]$ and $\bg_{2i}[n]$, $i=1,\cdots,\Nr$ can be approximated as independently and identically distributed complex Gaussian random variables with mean $0$ and variance $1$. Consequently, $\theta(n)$ is the angle between two independent random vectors, and $|\sin(\theta(n))|^2$ is beta distributed with parameters $\Nr-1$ and $1$.

When the number of receive antennas $\Nr =2$, $|\sin(\theta(n))|^2$ is uniformly distributed on $[0,\;1]$.
\end{proof}
Thus, a closed form expression of the complementary probability of outage
$\mathbb{P}\left(\rmR[n] \geq \rmR^{\rmt}[n-\rmD]\right)$ can be written as
\begin{eqnarray}
\nonumber\mathbb{P}\left(\rmR[n] \geq \rmR^{\rmt}[n-\rmD] \;|\bH_1[n-\rmD]\right) &=&
\mathbb{P}\left(\gamma[n] \geq \frac{\|\bH_1[n-\rmD]\bff_1[n-\rmD]\|^2}{p}\;\large|\;\bH_1[n-\rmD]\right)\\
\nonumber &=& \mathbb{P}\left(|\sin(\theta_n)|^2 \geq \frac{\|\bH_1[n-\rmD]\bff_1[n-\rmD]\|^2}{p\|\bH_1[n]\bff_1[n-\rmD]\|^2}\;\large|\;\bH_1[n-\rmD] \right)\\
&=& 1-\left(\frac{\|\bH_1[n-D]\bff_1[n-D]\|^2}{p\|\bH_1[n]\bff_1[n-D]\|^2}\right)^{\Nr-1}.
\end{eqnarray}
Consequently, the ergodic system goodput for a delay D is given by
\begin{equation}\label{eqn:rd}
\bar{\rho}(D) = \sum_{k_{10},k_{1D}}{R^{t}_{k_{1D}}\left(1-\left(\frac{\|\bH_1[n-D]\bff_1[n-D]\|^2}{p\|\bH_1[n]\bff_1[n-D]\|^2}\; \large|\;\bH_1[n] \in \mathcal{V}_{k_{10}}, \bH_1[n-D] \in \mathcal{V}_{k_{1D}} \right)^{\Nr-1}\right)\left[\mathrm{P}^D\right]_{k_{1D}k_{10}}\pi}.
\end{equation}
Finally, the ZF feedback goodput gain is readily written, following equations (\ref{eqn:rd}) and (\ref{eqn:thm_conv}).
\begin{proposition}\label{prp:prop_ZF}
 For fixed feedback delay of D samples, the feedback throughput gain with ZF
interference nulling of the strongest interferer can be bounded as
\begin{eqnarray}\label{eqn:boundZF}
\nonumber\Delta\bar{\rho}(\rmD)&=&\sum_{k_{10},k_{1D}}{\rmR^{\rmt}_{k_{1D}}\left(1-\frac{\|\bH_1[n-\rmD]\bff_1[n-\rmD]\|^2}{p\|\bH_1[n]\bff_1[n-\rmD]\|^2}  \right)\left(\left[\rmP^{\rmD}\right]_{k_{1D}k_{10}}-\pi\right)\pi}\mathrm{  }\\
&\leq& \rmc \left(\lambda\right)^{\rmD},
\end{eqnarray}
where $\rmc = \sum_{k_{1D}}\sqrt{\pi_{k_{1D}}}\max_{k_{10}}{\left(\rmR^{\rmt}_{k_{1D}}\left(1-\left(\frac{\|\bH_1[n-\rmD]\bff_1[n-\rmD]\|^2}{p\|\bH_1[n]\bff_1[n-\rmD]\|^2}\right)^{\Nr-1}\right)\right)}$,
and $\lambda$ is the second largest eigen value of the matrix $\mathbf{P}\tilde{\mathbf{P}}$.
\end{proposition}

The following remarks are in order.
\begin{enumerate}
\item Employing ZF interference cancellation at the receiver brings back the exponential rate of decrease of the ergodic goodput to $\sqrt{\lambda}$, similar to the case reported in \cite{LFTCJ}, where the system is noise limited. The authors in \cite{LFTCJ}, however, derive expressions for the ergodic throughput of the system, without taking into account the rate outage. We reformulate the result in \cite{LFTCJ} in terms of the goodput metric as follows
    \begin{eqnarray}\label{eqn:boundNS}
    \Delta\bar{\rho}(\rmD)&=&\sum_{k_{10},k_{1D}}{\rmC_{\rmN}(\rmD)\left(\left[\mathrm{P}^D\right]_{k_{0D}k_{00}}-\pi\right)\pi}\mathrm{ }\leq \kappa \left(\lambda\right)^{\rmD},
\end{eqnarray}
where $\rmC_{\rmN}(\rmD) = \rmR^{\rmt}_{k_{1D}}\mathbb{P}(\rmR_{k_{1D}k_{10}}\geq \rmR^{\rmt}_{k_{1D}})$, and
$$\rmR_{k_{1D}k_{10}} = \log_2\left(1+ \frac{\alpha_1\|\bH_1[n]\bff_1[n-\rmD]\|^2}{\Nt\|\bw^*[n]\bv_1[n]\|^2}\;\large|\;\bH_1[n]\in \mathcal{V}_{k_{10}}, \bH_1[n-\rmD]\in \mathcal{V}_{k_{1D}}\right),$$
and $\kappa = \sum_{k_{1D}}{\sqrt{\pi}\max_{k_{10}}{\rmC_{\rmN}(\rmD)}}$.
\item The coefficient c depends on the effective channels at the transmitter and the receiver, and it is
smaller than the coefficient $\kappa$ in the upperbound of (\ref{eqn:boundNS}). This makes sense because the projection of the effective channel at the receiver over the subspace perpendicular to that spanned by the interference vectors, causing an effective power loss in the SNR at the receiver.
\item In Section \ref{sec:results}, we run simulations to show the tightness of the derived upper bounds. We compare the performance of the system with and without interference cancellation, with that of the noise limited environment.
\end{enumerate}
%
\section{Precoded Spatial multiplexing}\label{sec:Precoding}
Spatial multiplexing can offer higher data rates by sending multiple data streams to the receiver. Limited feedback precoded spatial multiplexing is included in the emerging 3GPP-LTE standard. The performance of this MIMO system, however, in the presence of feedback delay is yet to be analyzed.
In this section, we extend the goodput analysis in Section \ref{sec:analysis} to precoded spatial multiplexing systems, and we evaluate its performance limits in the presence of feedback delay and other cell interference.

For a precoding matrix $\bF$ with $\Ns$ transmit streams, the achievable rate at $M_1$ is computed as
\begin{equation}\label{eqn:RRXprec}
\rmR[n] = \log_2\left(\det\left(\bI_{\Nr}+ \bK_1[n](\bK_{I}[n])^{-1}\right)\right),
\end{equation}
where $\bK_1[n] = \frac{\alpha_1}{\Nt N_0}\bH_1[n]\bF_1[n-\rmD]\left(\bH_1[n]\bF_1[n-\rmD]\right)^*$ denotes the covariance matrix of the desired signal from $B_1$
$$\bK_{I}[n] = \bI_{\Nr} + \frac{\alpha_2}{\Nt N_0}\bG_2[n]\bF_2[n-\rmD](\bG_2[n]\bF_2[n-\rmD])^*$$ is the covariance matrix of the interference from $B_2$ plus the noise.

The rate at the base station is computed without taking into account the presence of the noise and the interference at the mobile, as in Section \ref{sec:analysis},
\begin{equation}\label{eqn:RTXprec}
\rmR^{\rmt}[n-\rmD] = \log_2\left(\det\left(\bI_{\Nr}+ \bH_1[n-\rmD]\bF_1[n-\rmD](\bH_1[n-\rmD]\bF_1[n-\rmD])^*\right)\right).
\end{equation}
The goodput then follows as
\begin{equation}
\rho_{\rmP}(\rmD)=\rmR^{\rmt}[n-\rmD]\cI\left(\rmR^{\rmt}[n-\rmD] \leq \rmR[n]\right).
\end{equation}
The ergodic goodput is computed in terms of the joint probability mass functions of the quantized channels $\cH$.
The quantization of the channel state information according to the capacity maximization selection criterion, as explained in Section \ref{sec:LFfeedback}, maps the channels into a first order Markov chain model, where the different states correspond to different Voronoi regions. Thus, the same analysis can be applied to computing the ergodic goodput gain for the precoded spatial multiplexing system.

Define $\rmC_{\rmP}(\rmD)$ as the instantaneous goodput at time $n$ with $\bH_1[n-D] \in \cV_{k_{1D}}$, $\bH_1[n] \in \cV_{k_{10}}$, $\bH_2[n] \in \cV_{k_{20}}$ and $\bH_2[n-\rmD] \in \cV_{k_{2D}}$. The ergodic goodput gain
$\Delta\bar{\rho}_{\rmP}(\rmD)$ can thus be obtained using equation (\ref{eqn:thm_conv}) and Proposition \ref{prp:prop_1}.
\begin{proposition}\label{prp:prop_precoding}
For fixed feedback delay of D samples, the ergodic goodput gain $\bar{\rho}_{\rmP}$ is upper bounded as
\begin{equation}
\Delta\bar{\rho}_{\rmP}(\rmD) \leq a\left(\sqrt{\lambda}\right)^{2\rmD} + b\left(\sqrt{\lambda}\right)^{\rmD}.
\end{equation}
where $a =  r\sum_{k_{1D},k_{2D}}{\max_{k_{10},k_{20}}{\mathrm{C}_{\rmP}(\rmD)}\pi}$, and
$$b = (1-r)\sum_{k_{1D},k_{10},k_{2D}}{\max_{k_{20}}{\mathrm{C}_{\rmP}(\rmD)}\pi^{2}\sqrt{\pi}}+ \sum_{k_{1D}, k_{2D},k_{20}}{\max_{k_{10}}{\mathrm{C}_{\rmP}(\rmD)}\pi^{2}\sqrt{\pi}}.$$
\end{proposition}
The coefficients $a$ and $b$ depend on the instantaneous rate $\rmC_{\rmP}(\rmD)$, the stationary probability distribution $\pi$ of the Markov chain, as well as the codebook size $N$.
The achievable rate of the precoded spatial multiplexing is dependent on the feedback delay and the interference from the neighboring channels. The decay rate depends on the amount of delay, the strength of the interference and the number of quantization levels. In Section \ref{sec:results}, we present numerical results to show the performance of the precoded spatial multiplexing system with respect to delay.

\section{Simulation Results}\label{sec:results}
In this section, we present simulation results to evaluate the performance of
limited feedback beamforming MIMO systems with delay in the presence of
other cell interference, for several scenarios.

For our simulations, we assume that the scattering environment is uniform such that the channel
coefficients are $\mathcal{CN}(0,1)$ and the temporal correlation follows Clarke's model and is
characterized by the continuous first order Bessel function $r(\tau) = \mathcal{J}_0(2\pi f_d\tau)$, where
$f_d$ is the maximum Doppler frequency and $\tau$ is the time separation between the samples. The discrete time counterpart of the continuous time autocorrelation function is implemented using the inverse Discrete Fourier Transform method proposed in \cite{BeaulieuRayleighMethod}, such that the discrete time samples have autocorrelation $r(n) = \mathcal{J}_0(2\pi\frac{f_d}{f_s}n)$, where $n$ is the sample separation, and $f_s$ is the sampling rate, and both the continuous and the discrete correlation functions have the same power spectrum.

\subsection{Limited feedback beamforming with other cell interference and feedback delay}
We compare the feedback rate gain of the interference limited system with that of the noise
limited system, for the limited feedback beamforming presented in Section \ref{sec:analysis}.
Figure \ref{fig:cpm_intf_nointf} plots both the throughput and goodput gains of the systems
versus the feedback delay D. For the throughput gain, we use the achievable rate at the receiver given by
\begin{equation}
\rmR[n] = \log_2\left(1 + \frac{\alpha_1\|\bH_1[n]\bff_1[n-\rmD]\|^2}{{\alpha_2\|\bw^*[n]\bG_2[n]\bff_2[n-\rmD]\|^2} + N_t\|\bw^*[n]\bv[n]\|^2}\right),
\end{equation}
i.e., without including the rate outage. For the noise limited environment,
the result for the system throughput goes back to that in \cite{LFTCJ} where the exponential rate of
decrease is $\sqrt{\lambda}$.

Figure \ref{fig:cpm_intf_nointf} shows the effect of the feedback delay on both the throughput
gain $\Delta\bar{\rmR}(\rmD)$, as well as the goodput $\Delta\bar{\rho}(\rmD)$, for a four transmit, four receive antenna system with a normalized Doppler shift of $0.025$.  We observe that including
the rate outage at the transmitter causes a decrease in the throughput gain of the system. The decrease is constant with respect to the delay.  This can be explained by the fact that the outage does not affect the Markov chain statistical properties, which mainly depend on the value of the delay. The gap between the
goodput and the throughput gain can be observed
for both the noise limited and interference limited scenarios.

Moreover, the results in Figure \ref{fig:cpm_intf_nointf} show that the exponential rate of decrease of the interference limited system (one cell interference) is more pronounced than that of the noise limited system (single-cell) for the cases of small to moderate delay (up to D $= 15$). The simulation results agree with the analytical result in Proposition \ref{prp:prop_1}. We also observe from Figure \ref{fig:cpm_intf_nointf} that the rate of decrease in the two-cell environment increases with the delay. This is explained by the fact that for $\lambda\leq 1$, the $\sqrt{\lambda}$ term in the right hand side of equation (\ref{eqn:bound}) becomes more dominant as D increases.

To evaluate the accuracy of the upper bound derived in Proposition \ref{prp:prop_1} in predicting the performance of the system, Figure \ref{fig:cmp_upperbound} plots the normalized feedback rate gain and the expression in (\ref{eqn:bound}), as a function of feedback delay D, for a normalized Doppler shift, $f_dT_s = 0.025$ . As can be observed from the Figure, the throughput gain is tightly upper bounded by (\ref{eqn:bound}), for delay intervals of interest. This upper bound is different from that in the noise limited scenario, plotted in the Figure for comparison.
This allows the use of the results of Proposition \ref{prp:prop_1} for computing gains in achievable rate for delay values of interest.

\subsection{ZF cancellation at the receiver}
The ZF linear receiver was proposed in Section \ref{sec:InterfCancel} to reduce the effect of other cell interference.
Figure \ref{fig:GoodputZFNSOCI} presents the effect of employing the ZF receiver architecture on the feedback throughput gain. We plot the ergodic goodput achieved by the noise limited (single-cell) MIMO transceiver, with that of the interference limited system, with and without interference cancellation. We observe that adding interference cancellation brings the decay rate of the feedback throughput gain back to $\sqrt{\lambda}$, similarly to the noise limited environment. The effective power loss due to the projection of the desired MIMO channel onto the interference subspace, causes the achievable rate to be less than that of the noise-limited environment.

Figure \ref{fig:GoodputZF} compares the normalized feedback goodput gain $\Delta\bar{\rho}(\rmD)$, and its approximation with the closed form upperbound in Proposition \ref{prp:prop_ZF}, for different values of the feedback delay D. The exact and closed form expression for the $\Delta\bar{\rho}(\rmD)$ are shown in solid and dashed lines respectively. The approximate $\Delta\bar{\rho}(\rmD)$ solution simulated based on the probability of outage is also shown in the Figure. One can observe that the closed form expression closely matches the exact $\Delta\bar{\rho}(\rmD)$.


\subsection{Precoded spatial multiplexing with feedback delay and other cell interference}
We investigate the performance of the precoded spatial multiplexing system in the presence of other cell interference in Figure \ref{fig:cpm_precod}. We consider a four transmit, four receive antenna MIMO system, with two spatial multiplexing streams $N_s = 2$ and a codebook size of 16. Figure \ref{fig:cpm_precod} plots the throughput gain for the precoded multistream system for both the interference limited and the noise limited environments, for a Grassmannian codebook of size 16\footnote{The codebooks used in these simulations are available at \url{http://cobweb.ecn.purdue.edu/~djlove/grass.html}} . We observe from the figure that the goodput gain of the precoded multistream decreases exponentially with the feedback delay. This exponential decrease is more pronounced in the presence of other cell interference, as predicted in Proposition \ref{prp:prop_precoding}.

Figure \ref{fig:cpm_precod_beam} plots the throughput gain of the precoded spatial multiplexing system and that of transmit beamforming, for the same codebook. The behavior of both systems vis-a-vis other cell interference and delay is the same. However, the throughput gain achieved by the precoded system is higher than that of transmit beamforming, due to the multiple transmit streams.

\subsection{Effect of doppler frequency and codebook size}
The performance of the two-cell system in the presence of delay with varying codebook sizes is shown in Figure \ref{fig:cmp_codebook_sizes} for a two transmit, two receive MIMO system and codebook sizes $4$, $8$ and $16$. One can clearly see that the rate of decay of the goodput gain increases as the codebook size increases, it approaches that of the single cell system for larger codebook sizes. This suggests a tradeoff between the feedback rate and the goodput gain in interference limited scenarios.

In Figure \ref{fig:comp_shifts}, we plot the feedback goodput gain versus the feedback delay for different Doppler frequencies normalized by the sampling period $T_s$, $f_dT_s =\{20\times10^{-3}, 25\times10^{-3}, 50\times10^{-3}, 100\times10^{-3}\}$. We consider a four transmit, four receive MIMO system with a Grassmannian codebook size of $16$. Clearly, the exponential rate of decrease of the feedback gain is sensitive to the Doppler shift. As the Doppler shift, or the velocity of the mobile is increased, the feedback gain is decreased, resulting in a steeper decrease rate with respect to the feedback delay.
This can be explained as follows. Over a fixed delay, higher Doppler causes larger channel variation, which accelerates
the decrease of the feedback capacity gain with the feedback delay and hence results in a steeper curve slope.

\subsection{3GPP-LTE design example}
In this subsection, we present a design example that demonstrates the application of the results in this paper to designing limited feedback system for 3GPP LTE-Advanced standard \cite{3GPPLTEAUTRA}. The standard proposes using limited feedback beamforming and precoded spatial multiplexing over orthogonal frequency division multiple access (OFDMA). For the downlink OFDMA, the available frequency bandwidth is partitioned into frequency slots, called frequency subbands, and assigned to different users based on a given scheduling algorithm. Each subband consists of several orthogonal frequency division multiplexing (OFDM) symbols.  Limited feedback is performed on every OFDM symbol to increase its data rate.
Channel estimation is performed using pilot symbols located at the center of the subbands. One CSI feedback link is required for every subband, and feedback is performed every two subframes.

\begin{table}[h]
\centering{
\caption{LTE design Specifications}
  \label{tb:tb_specs}
\begin{tabular}{|c|c|}
  \hline
  Carrier Frequency & 2 GHz\\
  Bandwidth & 10 MHz \\
  Antenna Array & $\Nt = \Nr = 4$ \\
  Subband Bandwidth & 1.08 MHz \\
  Subframe (TTI) length & 1 msec \\
  CSI Feedback (2 TTI) & 2 msec \\
  Control Delay & 4 msec \\
  Mobility & 30 km/h\\
  CSI codebook size & N = 16\\
  \hline
\end{tabular}
}
\end{table}

The design specifications are summarized in Table \ref{tb:tb_specs}. The transmission bandwidth is
10 MHz at a carrier frequency of 2 GHz. Downlink transmission is organized into radio frames with a
frame duration of 10 ms.  We consider the frame structure applicable to frequency division duplex (FDD),
in which 10 subframes, of duration 1 ms each, are available for downlink transmission \cite{3GPPLTEAUTRA}.
At each scheduled user equipment (UE), the channel state information is fed back every 2 subframes, or 2 ms. The design of the precoder for 4 transmit antennas is based on the Householder transformation \cite{golub1996matrix}. The choice of the codebook is based  on its reduced computational complexity. The Householder precoding codebook satisfies a constant modulus where all the transmit antennas keep the same power level, regardless of which precoding matrix is used to maximize the power amplifier efficiency. Moreover, the nested property of the codewords permits the precoding matrix in a lower rank (such as with transmit beamforming) to be a submatrix of a higher rank precoding matrix (when sending multiple streams). The codebook also features a constrained alphabet that avoids the need for matrix multiplication \cite{juho2009mimo}.

The size of the codebook used for quantizing feedback CSI is 16. The mobility of the UE is up to 30 km/h. The doppler shift for the maximum speed of 30 km/h is of $f_m = 55.5$ Hz. Normalizing this shift by the subframe time of 1 ms, we get a normalized Doppler shift of $f_mT_s = 0.055$. The control delay is of 4 ms. The roundtrip ARQ delay is of 4 ms. Based on these requirements, we compute the goodput gain of the limited feedback precoding systems, using the LTE Householder codebooks\footnote{The codebooks used in these simulations are available at \url{http://users.ece.utexas.edu/~inoue/codebook/index.html}}, for realistic delay values of 4 and 6 ms.

The second largest eigenvalue used in Proposition \ref{prp:prop_1} is first computed by using the stochastic matrix of the feedback state markov chain, based on the codebook size and the normalized Doppler shift to be $\lambda = 0.7721$. The maximum feedback goodput gain for delay free CSI feedback is computed as 2.453 bps/Hz. By using Proposition \ref{prp:prop_1} for the LTE codebook, with $N_s = 1$ stream, the normalized goodput gain is computed at (4, 6) ms respectively as (0.4708, 0.2904), this implies that the goodput gain at these values is  (1.1549, 0.7124) bps per subcarrier, and (83.1528, 51.2928) bps per subband.


\section{Conclusion}\label{sec:Conclusion}
In this paper, we analyzed the effect of the feedback delay on limited feedback systems in the presence of uncoordinated other cell interference. We showed that the delay on the feedback channel causes the base station to transmit at a rate higher than the instantaneous mutual information of the channel, resulting in a rate outage. By only considering the successful transmissions over the channel, or the goodput, we showed that the other cell interference causes the exponential rate of decay of the ergodic feedback gain to double, when compared to the noise-limited single cell scenario, especially at low to moderate feedback delay values. The analysis was carried out for both transmit beamforming and precoded spatial multiplexing systems. Numerical results confirmed that a smaller time correlation and a lower codebook size leads to a faster reduction of the capacity gain with the feedback delay.

We then implemented interference cancellation, assuming that the receiver has perfect knowledge of the interference channel, while the base station is oblivious to the presence of any interference. We used the zero forcing receiver architecture, and we showed that the effect of the other cell interference can be simply mitigated, by bringing the decay rate of the throughput gain back to that of the noise limited case, at the expense of a rate loss due to the projection power loss of the ZF receiver.

This paper opens up several issues for future investigation. First, the results in this work focus on the feedback delay as the major
 bottleneck in the performance, these results can be extended to include other nuisance sources
 on the channel, such as errors in estimating the CSI at the receiver, and noise on the feedback channel. Future
 work should also include extending the results to multi-user MIMO scenarios and multiple access (MAC) channels.
 More importantly, this paper sets the grounds for implementing a joint rate adaptation and codebook design technique
 that takes the rate outage caused by the other cell interference as the major constraint, and optimizes the system
 goodput accordingly.
%
\appendices
\section{Proof of proposition (\ref{prp:prop_1})}\label{app:app1}
\textit{Proposition} \ref{prp:prop_1}:
For fixed feedback delay of D samples, the feedback goodput gain can be bounded as
\begin{equation}\label{eqn:bound}
\Delta\bar{\rho}(\rmD) \leq a\left(\sqrt{\lambda}\right)^{2\rmD} + b\left(\sqrt{\lambda}\right)^{\rmD},
\end{equation}
where $a =  r\sum_{k_{1D},k_{2D}}{\max_{k_{10},k_{20}}{\mathrm{C}(\rmD)}\pi}$, and
$$b = r\sum_{k_{1D},k_{10},k_{2D}}{\max_{k_{20}}{\mathrm{C}(\rmD)}\pi^{2}\sqrt{\pi}} + \sum_{k_{1D}, k_{2D},k_{20}}{\max_{k_{10}}{\mathrm{C}(\rmD)}\pi^{2}\sqrt{\pi}}.$$
%
\begin{proof}
We define
\begin{equation}
\mathrm{V}(D) = \sum_{k_{10},k_{1D}}{\sum_{k_{20},k_{2D}}{\mathrm{C}(\rmD)\left(\left[\mathrm{P}^{\rmD}\right]_{k_{1D}k_{10}}-\pi\right)\pi\left(
\left[\mathrm{P}^{\rmD}\right]_{k_{2D}k_{20}}-\pi\right)\pi}},
\end{equation}
and we expand it in terms of $\Delta\bar{\rho}_1(\rmD)$ as follows
\begin{eqnarray}
\nonumber \mathrm{V}(\rmD) &=& \Delta\bar{\rho}_1(\rmD) - (\sum_{k_{10},k_{1D}}{\sum_{k_{20},k_{2D}}{\mathrm{C}(\rmD)\left[\mathrm{P}^{\rmD}\right]_{k_{1D}k_{10}}\pi^3}} - \sum_{k_{10},k_{1D}}{\sum_{k_{20},k_{2D}}{\mathrm{C}(\rmD)\pi^4\mathrm{ }}})\\
\nonumber    &-& (\sum_{k_{10},k_{1D}}{\sum_{k_{20},k_{2D}}{\mathrm{C}(\rmD)\left[\mathrm{P}^{\rmD}\right]_{k_{2D}k_{20}}\pi^3}} - \sum_{k_{10},k_{1D}}{\sum_{k_{20},k_{2D}}{\mathrm{C}(\rmD)\pi^4\mathrm{ }}}).
\end{eqnarray}
Consequently, $\Delta\bar{\rho}(\rmD)$ is written as
\begin{eqnarray}
\nonumber \Delta\bar{\rho}(\rmD) &=& r\mathrm{V}(\rmD)\\
\nonumber     &+& (\sum_{k_{10}, k_{1D}}{\sum_{k_{20}, k_{2D}}{ \mathrm{C}(\rmD)\left[\mathrm{P}^{\rmD}\right]_{k_{1D}k_{10}}\pi^{3}}} - \sum_{k_0, k_{0D}}{\sum_{k_{2D}, k_{20}}{\mathrm{C}(\rmD)\pi^{4}}}\mathrm{ })\\
    &+& r(\sum_{k_{10}, k_{1D}}{\sum_{k_{20}, k_{2D}}{ \mathrm{C}(\rmD)\left[\mathrm{P}^{\rmD}\right]_{k_{2D}k_{20}}\pi^{3}}} - \sum_{k_{10}, k_{1D}}{\sum_{k_{20}, k_{2D}}{\mathrm{C}(\rmD)\pi^{4}}}\mathrm{ }).
\end{eqnarray}
Using (\ref{eqn:thm_conv}),
 $\rmV(\rmD)$ can be upperbounded as follows
\begin{eqnarray}
\nonumber \rmV(\rmD)&=& \sum_{k_{10}, k_{0D}}{\sum_{k_{20}, k_{2D}}{\mathrm{C}(\rmD)(\left[\mathrm{P}^{\rmD}\right]_{k_{1D}k_{10}}-\pi)(\left[\mathrm{P}^{\rmD}\right]_{k_{2D}k_{20}}-\pi)\pi^{2}}}\\
&\leq& \nonumber \sum_{k_{1D}}{\sum_{k_{2D}}{\max_{k_{10},k_{20}}{\mathrm{C}(\rmD)}\sum_{k_{10},k_{20}}(\left[\mathrm{P}^{\rmD}\right]_{k_{1D}k_{10}}-\pi)(\left[\mathrm{P}^{\rmD}\right]_{k_{2D}k_{20}}-\pi)
\pi^{2}}}\\
\nonumber &\leq&\sum_{k_{1D}}{\sum_{k_{2D}}{\max_{k_{10},k_{20}}{\mathrm{C}(\rmD)}\pi(\sqrt{\lambda})^{\rmD}(\sqrt{\lambda})^{\rmD}}}\\
&=& c_1(\sqrt{\lambda})^{2\rmD},
\end{eqnarray}
where $c_1 = \sum_{k_{1D}}{\sum_{k_{2D}}{\max_{k_{10},k_{20}}{\mathrm{C}(\rmD)}\pi}}$.

Similarly upper bounding the other terms yields expressions in $(\sqrt{\lambda})^{\rmD}$
\begin{eqnarray}
\nonumber \sum_{k_{10}, k_{1D}}{\sum_{k_{20},
k_{2D}}{\mathrm{C}(\rmD)(\left[\mathbf{P}^{\rmD}\right]_{k_{1D}k_{10}}-\pi)\pi^3}}
&\leq& \sum_{k_{1D},k_{20},k_{2D}}{\max_{k_{10}}{\mathrm{C}(\rmD)}\sum_{k_{10}}(\left[\mathbf{P}^{\rmD}\right]_{k_{1D}k_{10}}-\pi)\pi^3}\\
&\leq& \sum_{k_{1D},k_{20},k_{2D}}{\max_{k_{10}}{\mathrm{C}(\rmD)}\pi^{2}\sqrt{\pi}(\sqrt{\lambda})^{\rmD}}\\
&=& c_2(\sqrt{\lambda})^{\rmD},
\end{eqnarray}
where $c_2 =  \sum_{k_{1D},k_{20},k_{2D}}{\max_{k_{10}}{\mathrm{C}(\rmD)}\pi^{2}\sqrt{\pi}}$,
and
\begin{eqnarray}
\nonumber \sum_{k_{10}, k_{1D}}{\sum_{k_{20},
k_{2D}}{\mathrm{C}(\rmD)(\left[\mathbf{P}^{\rmD}\right]_{k_{2D}k_{20}}-\pi)\pi^3}}
&\leq& c_3(\sqrt{\lambda})^{\rmD}
\end{eqnarray}
with $c_3 = \sum_{k_{2D},k_{10},k_{1D}}{\max_{k_{20}}{\mathrm{C}(\rmD)}\pi^{2}\sqrt{\pi}}$.

Finally, the expression $\bar{\rho}(\rmD)$ can be upperbounded by
\begin{equation}
\bar{\rho}(\rmD) \leq a(\sqrt{\lambda})^{2\rmD} + b(\sqrt{\lambda})^{\rmD},
\end{equation}
where $a =  rc_1 = r\sum_{k_{1D},k_{2D}}{\max_{k_{10},k_{20}}{\mathrm{C}(D)}\pi}$, and
$$b =rc_3 + c_2 =  r\sum_{k_{1D},k_{10},k_{2D}}{\max_{k_{20}}{\mathrm{C}(\rmD)}\pi^{2}\sqrt{\pi}}+ \sum_{k_{1D}, k_{2D},k_{20}}{\max_{k_{10}}{\mathrm{C}(\rmD)}\pi^{2}\sqrt{\pi}}.$$
\end{proof}


%

\bibliographystyle{IEEEtran}
\bibliography{IEEEabrv,ReferenceFile, Takaoabrv}

%
%
\begin{figure}[h]
  \begin{center}
    \includegraphics[scale=0.35]{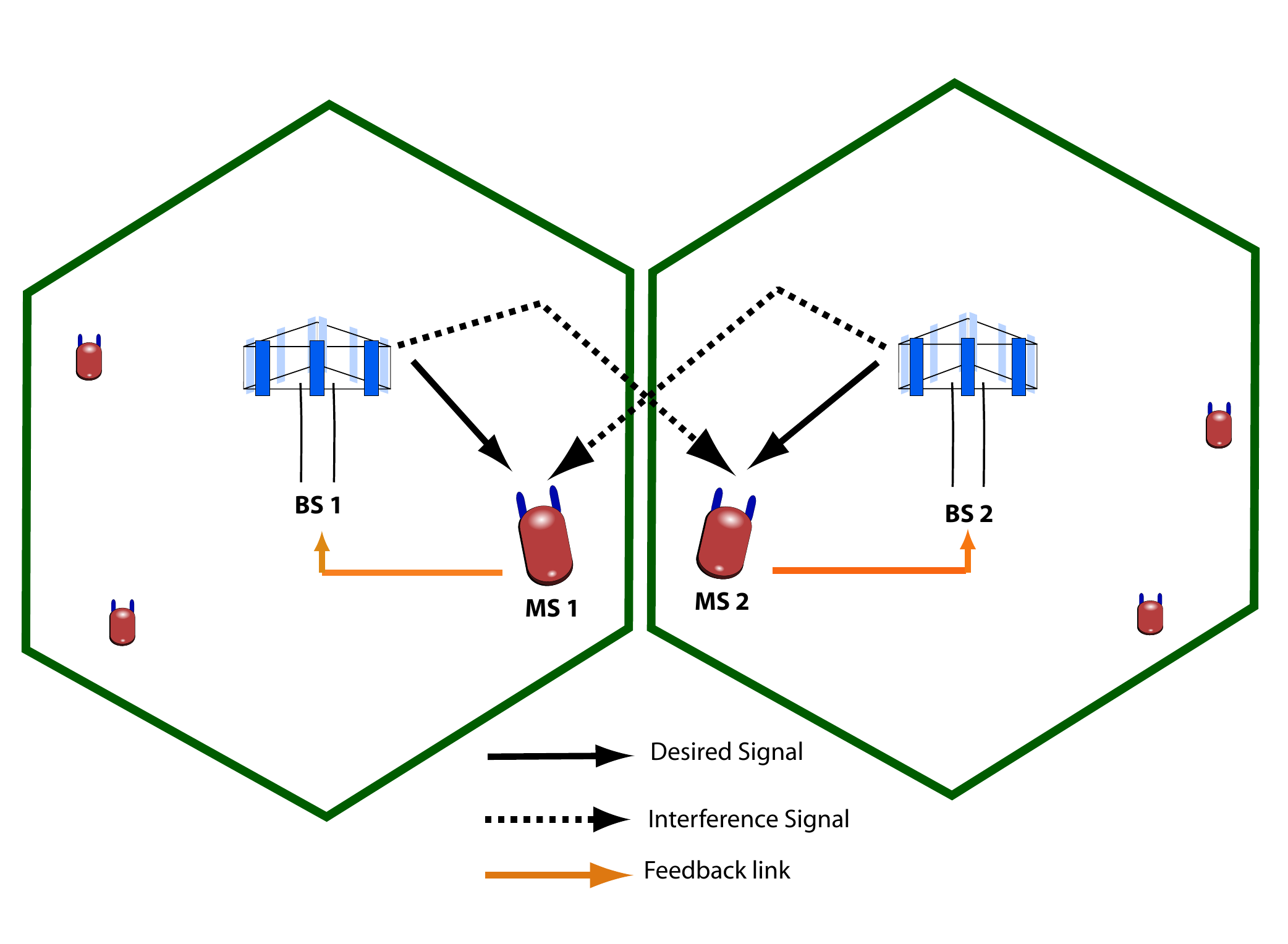}
    \caption{In the downlink scenario, the mobile user
    $MS_1$ experiences out of cell interference from $BS_2$. $\bG_2$ is the channel between $BS_2$ and $MS_1$, $\bH_1$ is the channel between BS1 and MS1. The beamforming vector associated with $\bG_2$ is $\bF_2$ corresponding to $\bH_2$ at $MS_2$.}
    \label{fig:cell_sys}
  \end{center}
\end{figure}
\begin{figure}[h]
  \begin{center}
    \includegraphics[scale = 0.35]{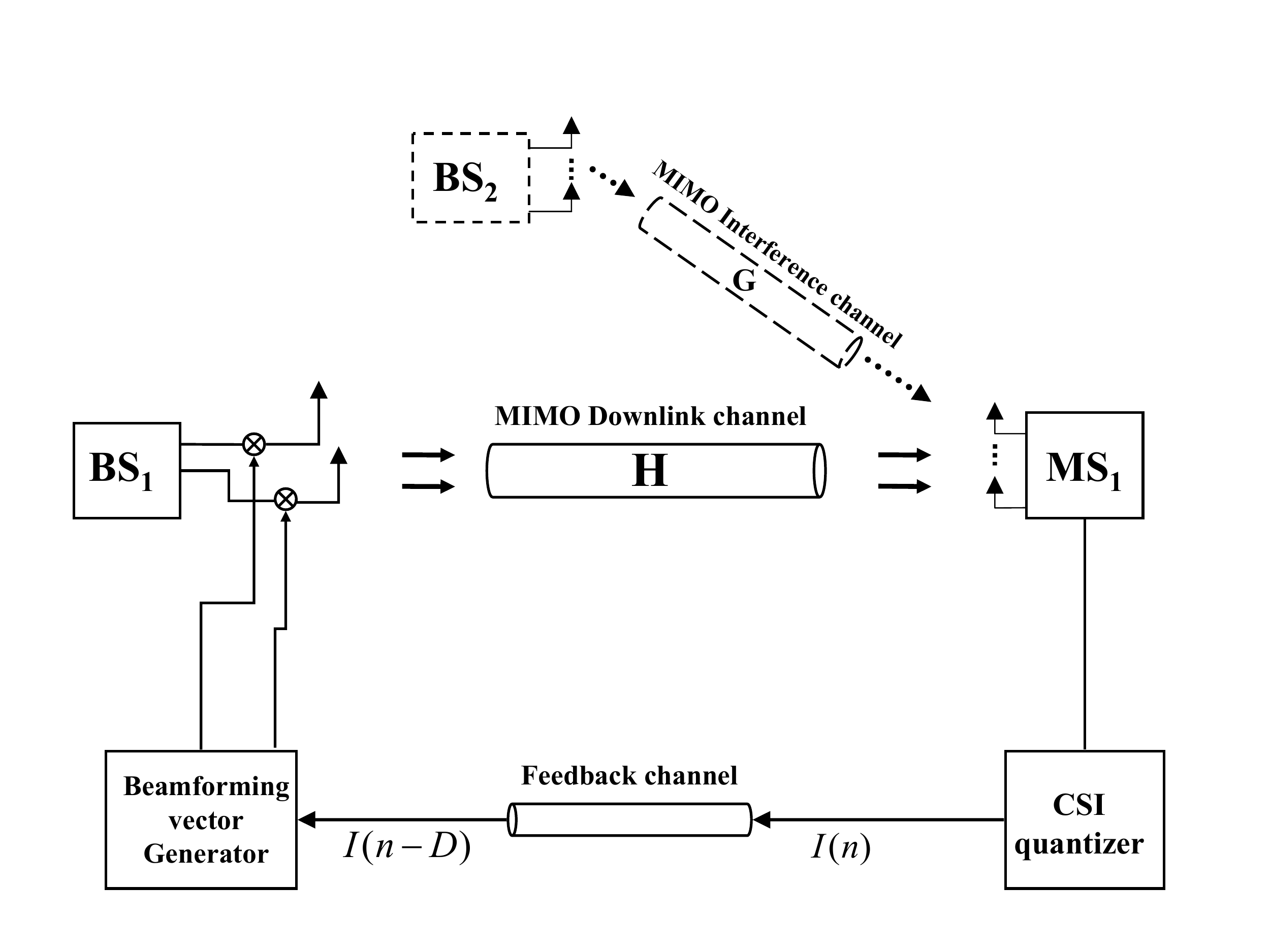}
    \caption{Limited feedback beamforming system with out of cell interference}
    \label{fig:LF_sys}
  \end{center}
\end{figure}
%


\begin{figure}[h]
 \begin{center}
   \includegraphics[scale = 0.35]{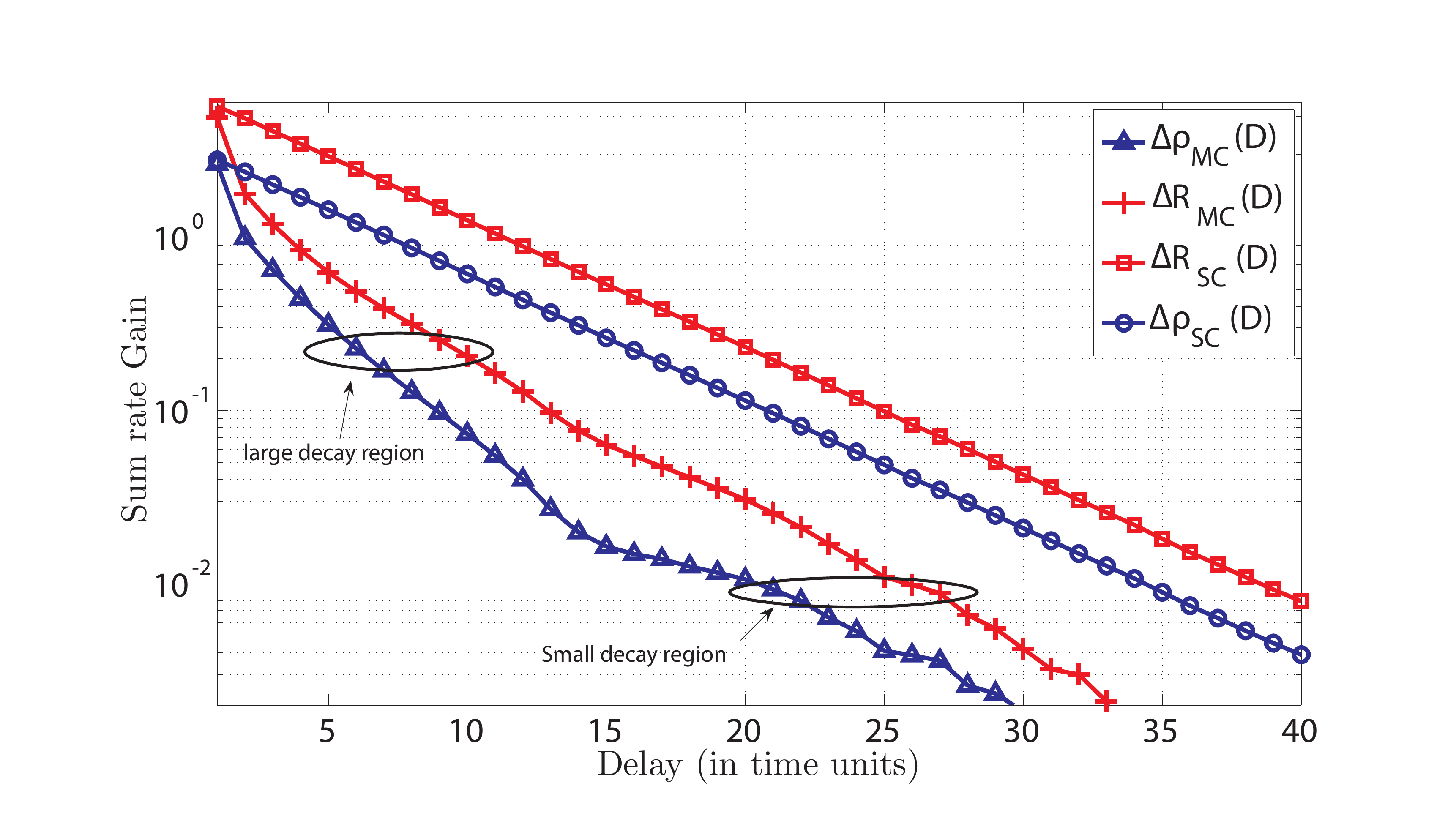}
    \caption{Effect of feedback delay on the ergodic sum rate gain.
    Comparison of an interference limited system versus a noise limited system for a $4\times 4$ MIMO system and a codebook size of $16$. The normalized Doppler shift is $f_dT_s = 0.025$. The subscript SC denotes single cell, while MC denotes multicell.}
    \label{fig:cpm_intf_nointf}
  \end{center}
\end{figure}


\begin{figure}[h]
 \begin{center}
   \includegraphics[scale = 0.35]{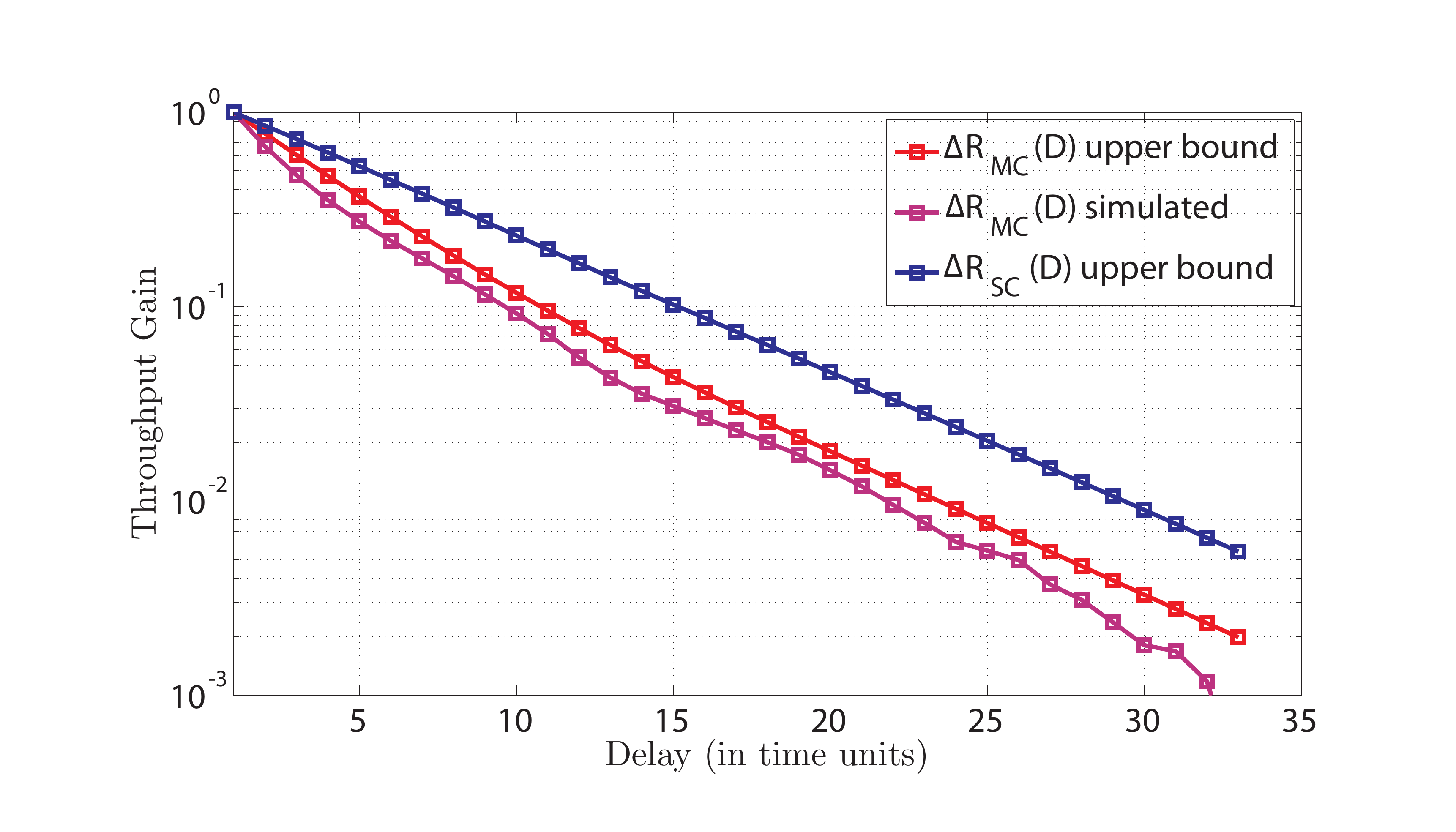}
    \caption{The normalized feedback throughput gain and its upperbound versus feedback delay for a Doppler shift $f_dT_s = 0.025$. The system considered is a four transmit, four receive MIMO system with a code size of 16. The normalized feedback throughput gain for the single cell scenario is given by $\Delta \rmR_{sc}(\rmD)$.}
    \label{fig:cmp_upperbound}
  \end{center}
\end{figure}
%
\begin{figure}[h]
  \begin{center}
    \includegraphics[scale=0.35]{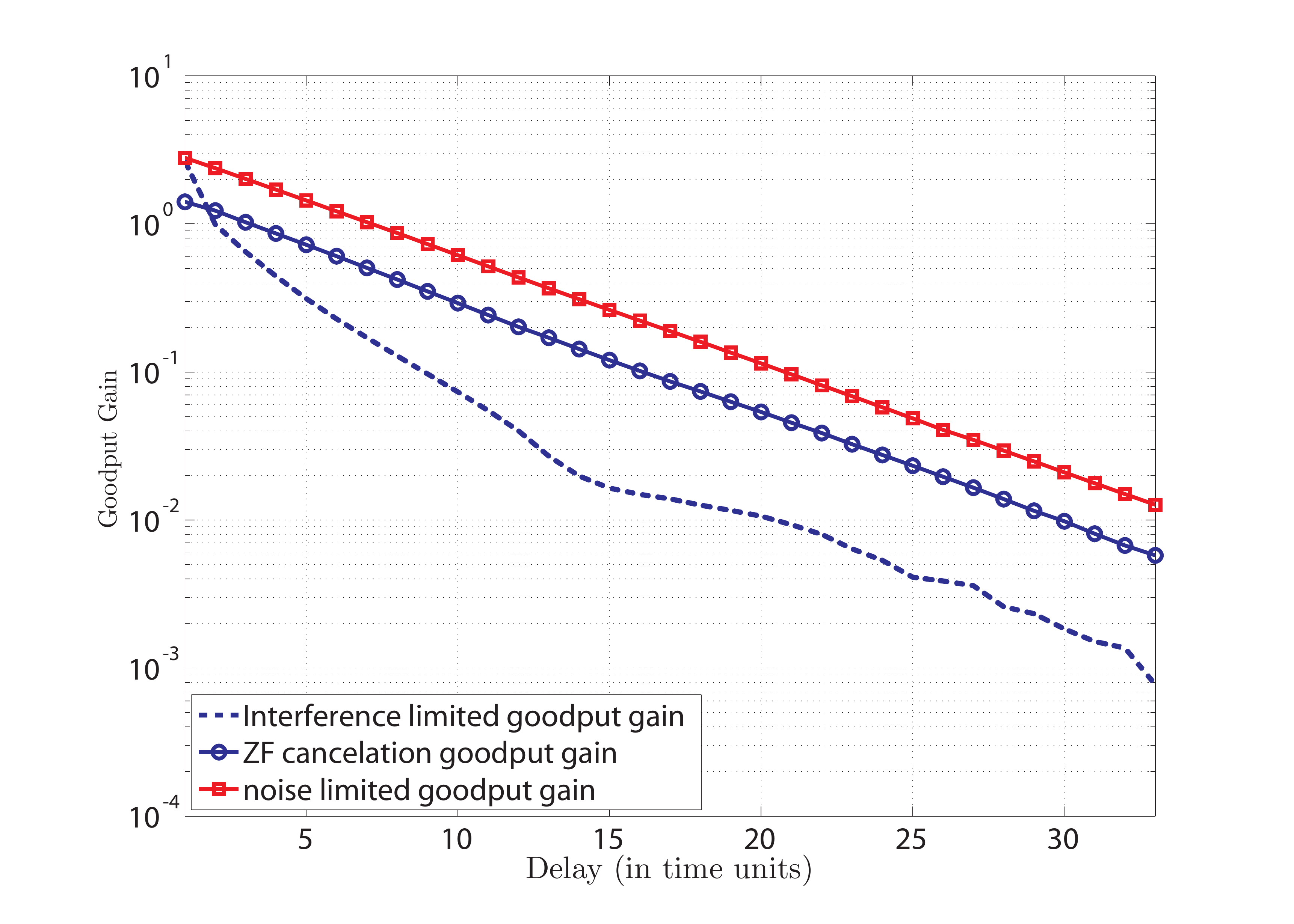}
    \caption{The feedback goodput gain versus feedback delay for the noise limited, interference limited and interference cancellation systems. The MIMO system considered is a $4\times 4$ MIMO system with a code size of $16$. The normalized Doppler frequency $f_dT_s = 0.025$.}
    \label{fig:GoodputZFNSOCI}
  \end{center}
\end{figure}
%

\begin{figure}[t]
  \begin{center}
    \includegraphics[scale=0.35]{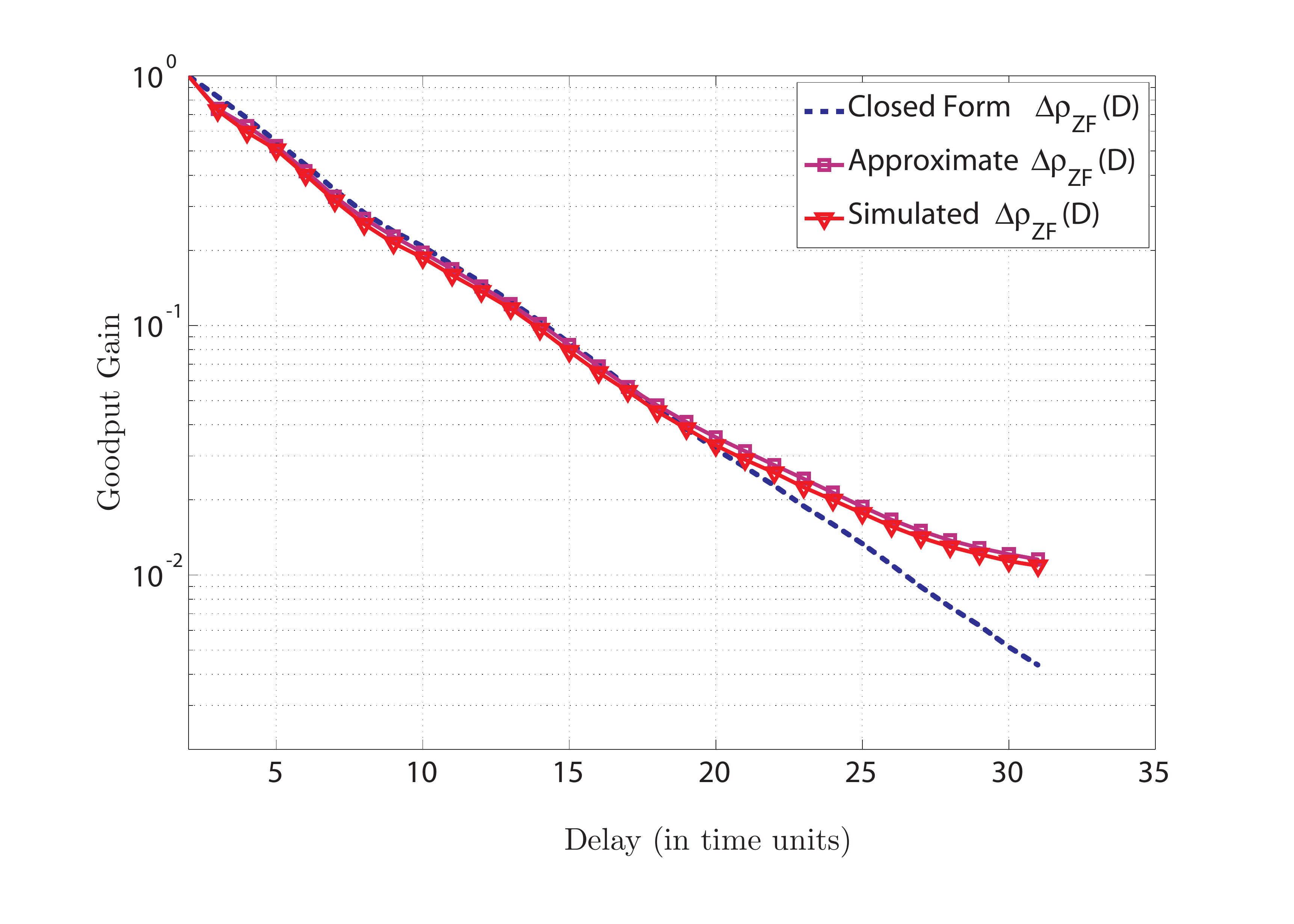}
        \caption{The feedback goodput gain and its approximations versus feedback delay for the ZF interference cancellation receiver. The system considered is a $4\times 4$ MIMO system with a code size of $16$.}
    \label{fig:GoodputZF}
  \end{center}
\end{figure}

\begin{figure}[h]
 \begin{center}
   \includegraphics[scale = 0.35]{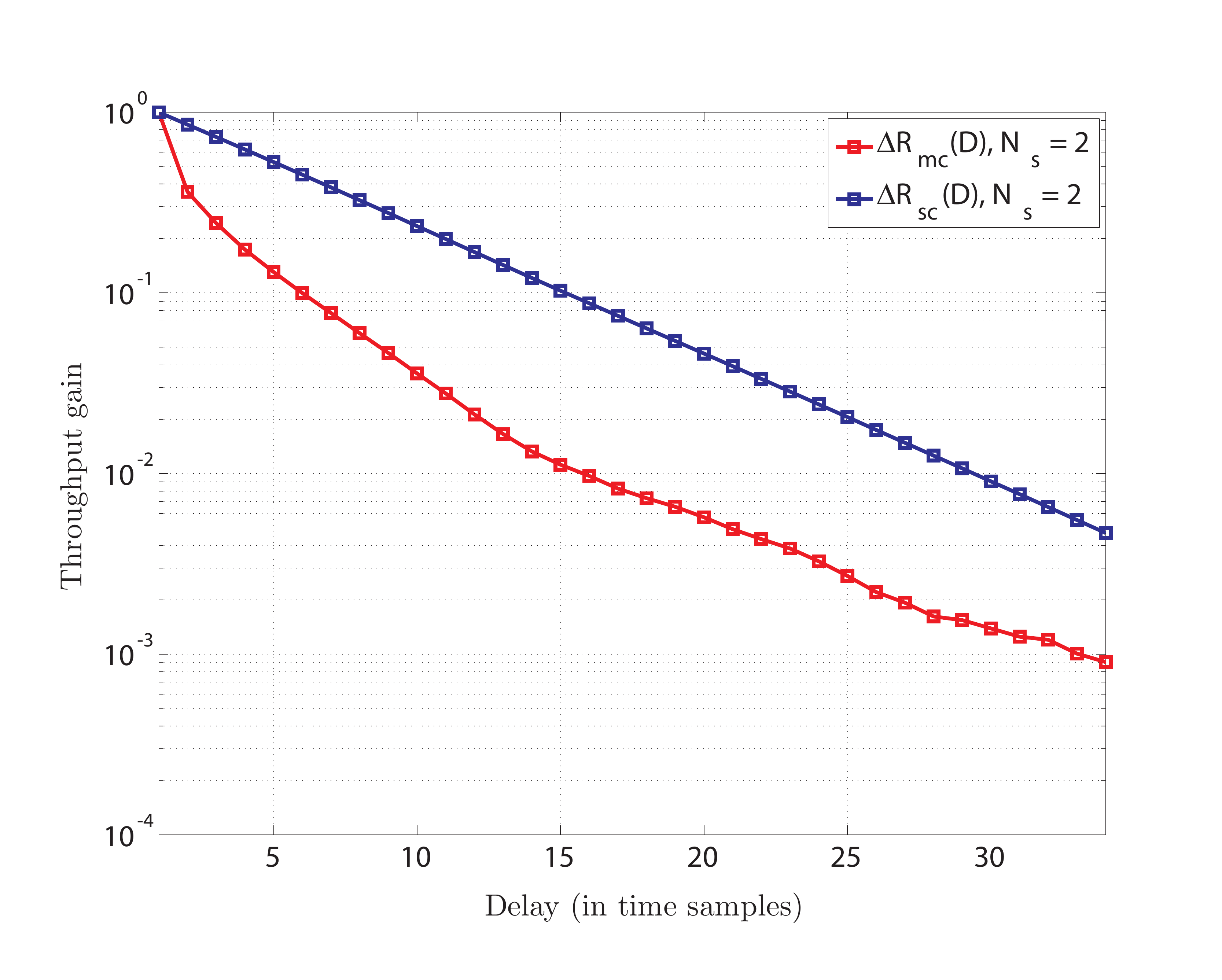}
    \caption{Effect of feedback delay on the ergodic sum rate gain for a precoded spatial multiplexing system.
    Comparison of an interference limited system versus a noise limited system for a $4\times 4$ MIMO system and a codebook size of $16$. The codebook used is a Grassmannian codebook with $N_s = 2$ spatial multiplexing streams.}
    \label{fig:cpm_precod}
  \end{center}
\end{figure}

\begin{figure}[h]
 \begin{center}
   \includegraphics[scale = 0.35]{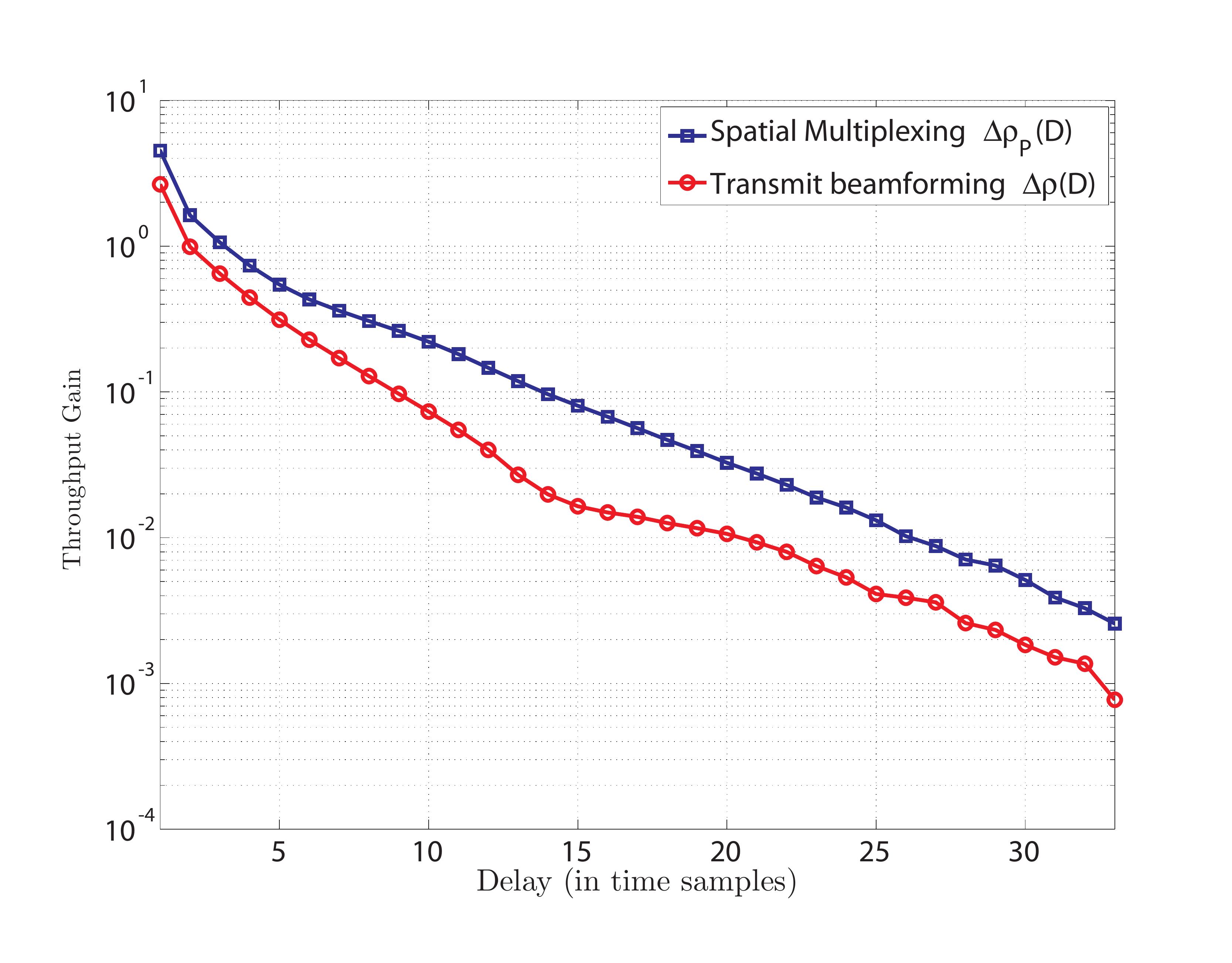}
    \caption{Comparison of the effect of feedback delay on the ergodic goodput gain for a precoded spatial multiplexing system and a transmit beamforming system.
    The system is interference limited with $4\times 4$ MIMO system and a codebook size of $16$. The codebook used is a Grassmannian codebook with $N_s = 2$ spatial multiplexing streams.}
    \label{fig:cpm_precod_beam}
  \end{center}
\end{figure}


\begin{figure}[t]
  \begin{center}
    \includegraphics[scale=0.35]{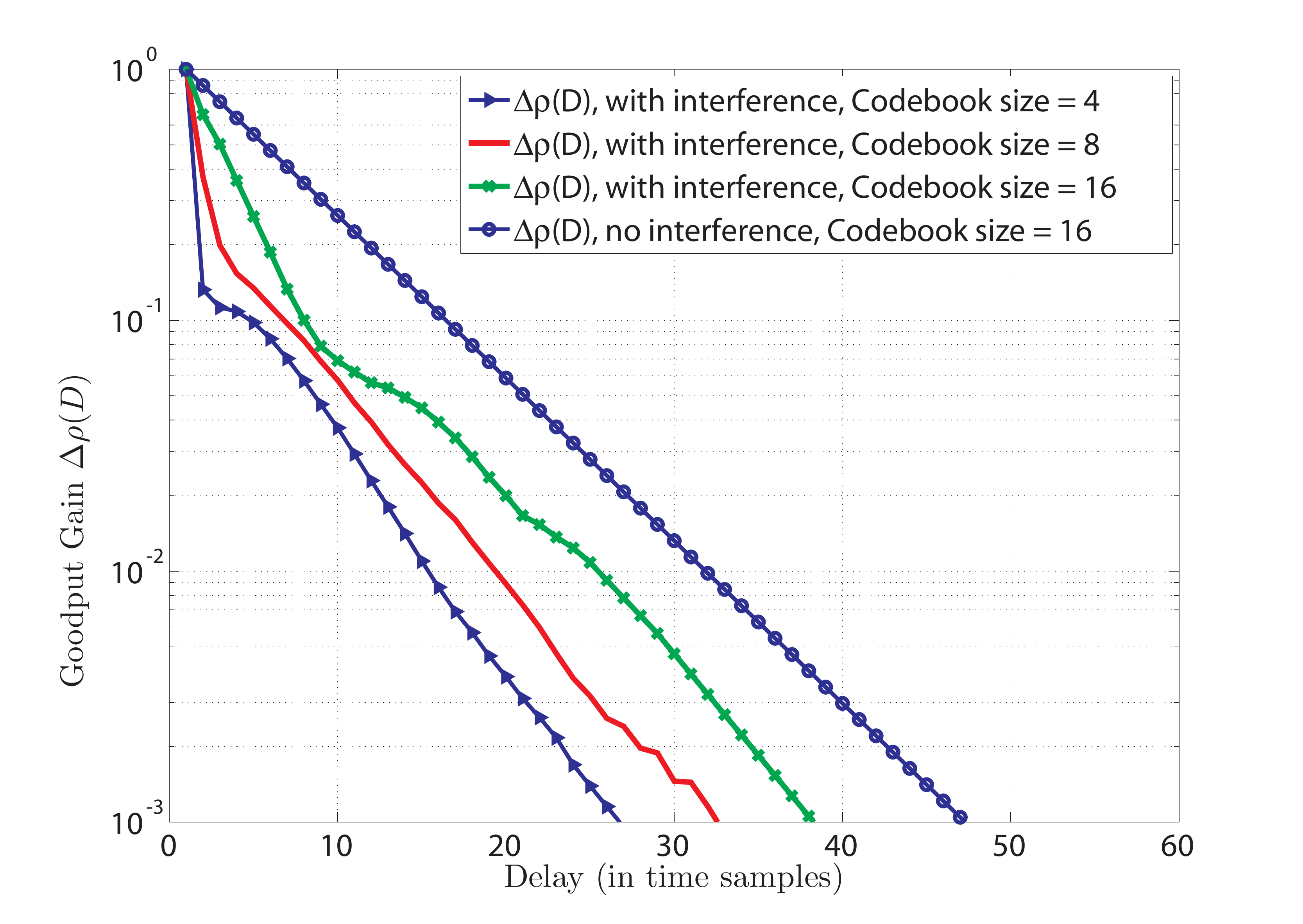}
    \caption{Effect of feedback delay on the feedback goodput gain of the beamforming system with limited feedback with varying codebook sizes. The system considered is a $2\times 2$ MIMO system with varying codebook sizes, and a normalized Doppler shift of $f_dT_s = 0.02$.}
    \label{fig:cmp_codebook_sizes}
  \end{center}
\end{figure}
%
%
%

%
\begin{figure}[t]
  \begin{center}
    \includegraphics[scale=0.35]{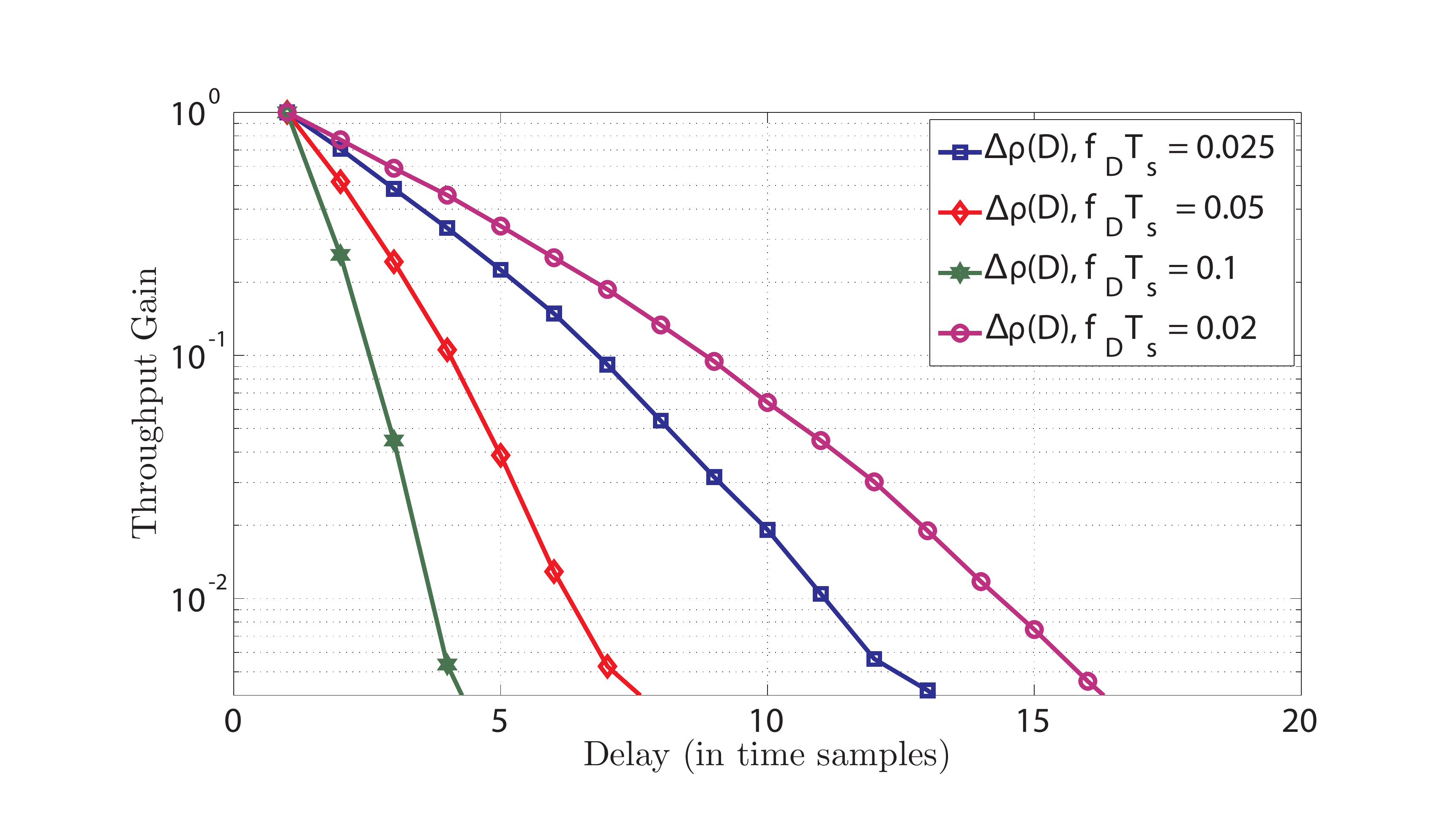}
    \caption{Effect of feedback delay on the feedback goodput gain of the beamforming system with limited feedback with varying Doppler shifts. The system considered is a $4\times 4$ MIMO system with a codebook size of 16.}
    \label{fig:comp_shifts}
  \end{center}
\end{figure}
%
%
%
%

\end{document}